\documentclass[aps,prx,twocolumn,superscriptaddress]{revtex4-1}
\usepackage{tikz}
\usepackage[english]{babel}
\usepackage[utf8]{inputenc}
\usepackage{indentfirst}
\usepackage{amsmath}
\usepackage{amssymb}
\usepackage{eufrak}
\usepackage{graphicx}
\usepackage{psfrag}

\newcommand{\abs}[1]{\left|#1\right|}
\newcommand{\complex}{\mathbb{C}}

\newcommand{\eps}{\varepsilon}

\newcommand{\ordo}{\mathcal{O}}
\newcommand{\veff}{v_{\text{eff}}}
\newcommand{\FQa}{\mathcal{Q}_\alpha}
\newcommand{\FJa}{\mathcal{J}_\alpha}
\newcommand{\FJab}{\mathcal{J}^\beta_\alpha}

\usepackage{amsthm}
\newtheorem{thm}{Theorem}

\newcommand{\vev}[1]{\left\langle #1 \right\rangle}
\newcommand{\ket}[1]{{\left|#1\right\rangle}}
\newcommand{\bra}[1]{{\left\langle #1\right|}}

\usepackage{ifpdf}

\ifpdf
\usepackage{epstopdf}
\usepackage[pdftex,colorlinks,urlcolor=blue,citecolor=blue,linkcolor=blue]{hyperref}
\else
\usepackage[hypertex,colorlinks,urlcolor=blue,citecolor=blue,linkcolor=blue]{hyperref}
\fi
\begin{document}
\numberwithin{equation}{section}

\title{  Current operators in Bethe Ansatz and Generalized Hydrodynamics:\\
  An exact quantum/classical correspondence
}
\author{M\'arton Borsi}
\affiliation{Department of Theoretical Physics, Budapest University
of Technology and Economics,\\ 1111 Budapest, Budafoki \'ut 8, Hungary}
\author{Bal\'azs Pozsgay}
\affiliation{Department of Theoretical Physics, Budapest University
  of Technology and Economics,\\ 1111 Budapest, Budafoki \'ut 8, Hungary}
\affiliation{BME Statistical Field Theory Research Group\\
1111 Budapest, Budafoki \'ut 8, Hungary}
\author{Levente Pristy\'ak}
\affiliation{Department of Theoretical Physics, Budapest University
of Technology and Economics,\\ 1111 Budapest, Budafoki \'ut 8, Hungary}

\begin{abstract}
  Generalized Hydrodynamics is a recent theory that describes
  large scale transport properties of one dimensional integrable 
  models.
It is built on the  (typically infinitely many) local conservation laws
present in these systems, and leads to a generalized Euler
type hydrodynamic equation.
Despite the successes of the theory, one
of its cornerstones, namely a conjectured expression for the currents
of the conserved charges in local equilibrium has not yet been
proven for interacting lattice models.
Here we fill this gap, and
compute an exact result for the mean values of current operators in
Bethe Ansatz solvable systems, valid in arbitrary finite volume. Our exact formula has a simple semi-classical
interpretation: the currents can be computed by summing over the
charge eigenvalues carried by the individual bare particles,
multiplied with an effective velocity describing their propagation in
the presence of the other particles. Remarkably, the semi-classical
formula remains exact in the interacting quantum theory, for any finite
number of particles and also in the thermodynamic limit. Our proof is
built on a form factor expansion and it is applicable to a large class
of quantum integrable models. 
\end{abstract}

\maketitle

\section{Introduction}

The description of the collective motion in
many body quantum systems is one of the most challenging
problems in theoretical physics. There are different possible levels for a theoretical
treatment, ranging from the microscopic laws to
various effective theories describing mesoscopic or macroscopic physics.
For large enough systems one expects that 
classical behaviour will emerge, at least for certain observables. 
It is thus important to understand how and under what circumstances the various classical
theories can
be derived from an underlying quantum mechanical motion \cite{Spohn-Book}. 

One such classical theory is hydrodynamics: it is known that many quantum systems admit some kind
of hydrodynamic description on the mesoscopic and/or macroscopic
scales \cite{Spohn-Book,lebowitz-hydro1}. Examples include Bose-Einstein
condensates \cite{BEC-book} or the quark gluon plasma
\cite{quark-gluon-hydro}. Superfluidity is a famous
exotic phenomenon, where frictionless flow is realized due to the
constraints for the decay of excitations into lower energy modes.
Superfluidity has been observed not only for
liquid Helium, but also in ultracold bosonic and fermionic gases
\cite{nature-quantum-revolution}.

An other class of systems with exotic hydrodynamic behaviour is comprised by the
one dimensional integrable models.
In these models there exist  independent conservation laws that constrain
the dynamical processes, the number of which grows at least linearly
with the volume.
As an effect, these models do not thermalize  to standard statistical
physical ensembles.
Instead, the emerging long time steady states can be described by a Generalized Gibbs Ensemble (GGE) that
involves all higher conserved charges of the model
\cite{rigol-gge,rigol-quench-review}.
The conservation laws prevent the decay of
quasi-particle excitations, and this leads to dissipationless and factorized
scattering, which was already demonstrated by experiments \cite{QNewtonCradle}.
Dissipationless
propagation of the collective modes leads to the emergence of ballistic transport and 
non-zero Drude weights (DC conductivity) \cite{Zotos-transport}. 

Generalized Hydrodynamics (GHD) is a recent theory describing large
scale non-equilibrium behaviour in  integrable models
\cite{doyon-ghd,jacopo-ghd} (see also
\cite{benjamin-takato-note-ghd,Doyon-GHD-LM,vasseur-ghd-1,vasseur-ghd-2}). The theory is built on
the local 
continuity equations following from the conservation laws, which lead
to a generalized Euler-type equation describing the ballistic
transport. The GHD provides exact results for the Drude weights
\cite{doyon-ghd-LL-Drude,jacopo-enej-ghd-drude}.
Diffusive corrections to the ballistic transport were also considered in
\cite{doyon-jacopo-ghd-diffusive,jacopo-benjamin-bernard--ghd-diffusion-long},
including an exact computation of the diffusion
coefficients. Remarkably, the predictions of the GHD
have already been confirmed in a concrete
experimental setup \cite{ghd-experimental-atomchip}.

Despite the successes of the GHD, one of the cornerstones of
the theory has
not yet been proven. The works \cite{doyon-ghd,jacopo-ghd}
conjectured an expression for the expectation values of the current
operators in local equilibrium, which is central to the derivation of
the main equations of motion.
Regarding integrable Quantum Field
Theories
a proof was provided in
\cite{doyon-ghd,dinh-long-takato-ghd},
whereas for the spin current of the XXZ model it was proven in
\cite{kluemper-spin-current}.
Nevertheless, for arbitrary current operators in interacting lattice
models or non-relativistic gases (the models most relevant to
experiments) it was completely missing up to now.
It is the goal of this paper to provide a proof of the
conjecture, valid for a wide class of  Bethe Ansatz solvable models.

The problem of the current mean values is also interesting from a
purely theoretical perspective, without the immediate 
application to GHD. A large body of literature has already addressed
equilibrium correlation functions in Bethe Ansatz solvable models
\cite{Korepin-Book,Jimbo-Miwa-book,KitanineMailletTerras-XXZ-corr1,QTM1,maillet-dynamical-1,maillet-dynamical-2,bjmst-fact6,2009JSMTE..04..003K,HGSIII,formfactor-asympt},
with particular interest devoted to the asymptotics of two-point
functions (for a review see \cite{karol-hab}), and to equilibrium mean
values of arbitrary short range operators of the Heisenberg spin
chains (for a review see
\cite{kluemper-goehmann-finiteT-review}). In contrast, the
current operators are very specific short range objects, and as we will
show, their finite volume mean
values take a remarkably simple form. This has not yet been noticed in
the Bethe Ansatz literature, and we believe that it deserves a study
on its own right.

In the following subsection we describe the conjecture of
\cite{doyon-ghd,jacopo-ghd} for the current mean
values and explain its role in the GHD, while omitting many technical
details. This brief introduction motivates our finite volume investigations.

After that, the remainder of this paper is composed as follows: In Section \ref{sec:main}
we specify the problem and present our main new results, with a
semi-classical interpretation given in \ref{sec:semi}.
Section \ref{sec:proof} includes our model independent proof, based on
a finite volume form factor expansion. 
Section
\ref{sec:construction} includes known generalities about the charge
and current
operators in integrable lattice models; also, it gives a short summary
of the Algebraic Bethe Ansatz.
The proof of our form factor expansion for the XXZ and XXX spin chains
is  presented in Section \ref{sec:fftcsa2}. In Section  \ref{sec:factorized}
we
point out a connection to the theory of factorized correlation
functions. Finally, we conclude in \ref{sec:conclusions}.

\subsection{Foundations of the GHD}

Isolated integrable models equilibrate to steady states described by
Generalized Gibbs Ensembles (GGE's). 
Each GGE can be
characterized by a set of parameters (generalized temperatures), or
alternatively, by the mean values of all the local and 
quasi-local conserved charges in the model
\cite{rigol-quench-review,prosen-enej-quasi-local-review,JS-CGGE}.  

The time scales of equilibration to the GGE are set by the microscopic laws. 
It follows that in mesoscopic or macroscopic dynamical processes
local equilibration happens much sooner than the characteristic times of
the transport processes. This leads to the hydrodynamic description:
GHD assumes
 the existence of fluid cells (regions in space much larger than the
inter-particle distance, and much smaller than the variation of the physical
observables) such that the state of each fluid cell can
be described by a local GGE. 
The parameters of these local GGE's are then space and time dependent.

The Generalized Eigenstate Thermalization Hypothesis (GETH)
\cite{rigol-geth} states that
in local equilibrium the local observables only depend on the mean
values of the conserved charges, and not on any other particular
detail of the states.
Thus, in order to describe the dynamical processes, it is
enough to establish flow equations for the conserved charges,
which will then determine all other physical observables on the
hydrodynamic scales.

For simplicity let us consider here a continuum model in the
thermodynamic limit. Let the complete
set of local and/or quasi-local conserved charges be
$Q_\alpha=\int_{-\infty}^\infty dx\ Q_\alpha(x)$, with $\alpha$ being
an 
index or a multi-index, and $Q_\alpha(x)$ being the charge density
operators. Conservation of the charges implies that there exist
current operators $J_\alpha(x)$ satisfying the equation of motion in
Heisenberg picture:
\begin{equation}
   \partial_t Q_\alpha(x,t)+\partial_x J_\alpha(x,t)=0.
\end{equation}
GHD concerns the mean values of these relations:
\begin{equation}
  \label{eqstate1}
  \partial_t \vev{Q_\alpha(x,t)}
  +\partial_x \vev{J_\alpha(x,t)}  =0.
\end{equation}
A closed set of flow equations can be obtained if the currents are
expressed using only local information about the charges. In hydrodynamics
this is performed in a derivative expansion:
\begin{equation}
  \begin{split}
   & \vev{J_\alpha(x,t)} =
    f_\alpha\Big[\{Q_\beta(x,t),\partial_xQ_\beta(x,t),\dots\}_{\beta=1,2,\dots}\Big].
     \end{split}
\end{equation}
In the first approximation
we neglect the spatial variations,
and express the currents in the fluid cells
using the mean values
$\vev{Q_\beta(x,t)}$ in that specific fluid cell only.
This approximation describes the ballistic part of the
transport.
The diffusive part of the transport can also be treated by considering the derivatives
$\partial_x\vev{Q_\beta(x,t)}$ \cite{jacopo-benjamin-bernard--ghd-diffusion-long}, but this will not
be considered here.

The ballistic flow equations will then follow from \eqref{eqstate1}, given that one can
compute the exact mean values of the currents in the local equilibrium
states, as a function of the charges,  or using any alternative description of
the local GGE's. This is the problem that we consider in this paper.

A large class of integrable models is solvable by the Bethe Ansatz \cite{Bethe-XXX};
prominent examples include the Heisenberg spin chains or the 1D Bose
gas with point-like interaction.
In these models the local equilibria in the thermodynamic limit can be characterized
by the root densities $\varrho_n(\lambda)$ of the interacting quasi-particles,
where $n$ stands for a particle type and $\lambda$ is the so-called
rapidity parameter. The root densities can be understood as
generalizations of momentum-dependent occupation numbers in free
theory. Dissipationless scattering implies that these densities are
well defined concepts even in the presence of interactions. The
construction of the GGE is equivalent to specifying all the
root densities $\varrho_n(\lambda)$, because they carry all
information about the local equilibria. This is the ultimate form
of the GETH, and it was understood in
\cite{JS-CGGE,enej-gge} following the earlier works
\cite{sajat-xxz-gge,essler-xxz-gge,JS-oTBA,sajat-oTBA,andrei-gge,sajat-GETH}. 

In GHD we thus need to specify the Bethe root densities for each fluid
cell, therefore they will depend also on the $x,t$ coordinates of the cell.
It is a very fruitful idea of the
GHD that instead of concentrating on the equations \eqref{eqstate1} for the
charges, one should derive flow equations for the rapidity distribution
functions. This can be achieved starting from \eqref{eqstate1},
by expressing both mean values
in the continuity equations using the densities $\varrho_n(\lambda)$ only.

The mean values of the charges can be computed
additively. In local equilibrium we have
\begin{equation}
  \label{Qa1}
  \vev{Q_\alpha}=\int d\lambda\ \varrho(\lambda) q_\alpha(\lambda),
\end{equation}
where $q_\alpha(\lambda)$ is the single particle eigenvalue of the
charges, and here we assumed only one particle species for
simplicity. In GHD \eqref{Qa1} is assumed to hold for each fluid cell separately.

For the currents it was conjectured in \cite{doyon-ghd,jacopo-ghd} that the mean values can be
computed using a semi-classical expression, namely by integrating over the
carried charge multiplied by an effective propagation speed:
\begin{equation}
  \label{Jtdl}
  \vev{J_\alpha}=\int d\lambda\ \varrho(\lambda) \veff(\lambda) q_\alpha(\lambda).
\end{equation}
Here the effective speed $\veff(\lambda)$ is a generalization of the
one-particle group velocity, which also takes into account the
interactions between the particles. It has a physical explanation
using a semi-classical argument: the one-particle 
wave packets suffer time delays due to the scattering on the other
particles, and these time delays
accumulate along the orbit, eventually modifying the propagation
speed. 
The resulting effective speed $\veff(\lambda)$ is a collective
property of the local GGE, because for each $\lambda$ it also 
depends on the particle density
$\varrho(\lambda')$ for all other $\lambda'$. For a precise definition
of 
$\veff(\lambda)$
we refer to \cite{doyon-ghd,jacopo-ghd}.

Eq. \eqref{Jtdl} has not yet been proven in Bethe
Ansatz. It is our goal to fill this gap and to prove \eqref{Jtdl}
starting from a rigorous finite volume computation.

Regarding the interpretation of \eqref{Jtdl} it was explained in
\cite{flea-gas}, 
that the GHD can be simulated by the
so-called ``flea gas'' model, which describes 1D motion of purely classical particles,
subject to time delays (displacements) as an effect of interparticle
collisions. Thus one observes a complete
quantum/classical correspondence on the hydrodynamic scale. 

In this work we show that the functional form of the mean values is the same in
finite and infinite volume, therefore the quantum/classical correspondence
holds with an arbitrary  finite number of particles.

\section{Current mean values}

\label{sec:main}

\subsection{Elements of integrability}

We consider integrable many body quantum models, including both lattice and
continuum theories \cite{Korepin-Book}. In this work we limit ourselves to those theories
where particle number is conserved, and which can be solved by the
traditional (``non-nested'') Bethe Ansatz.
Main examples are given by the Heisenberg spin chains, the 1D Bose
gas, and also certain integrable QFT's \cite{Korepin-Book}. Regarding
other types of integrable models we give a few comments in the Conclusions.

In this Section we present formulas pertaining to lattice
models.
We consider an integrable Hamiltonian $H$ in a finite volume $L$:
\begin{equation}
  H=\sum_{x=1}^L h(x).
\end{equation}
Here $h(x)$ is the Hamiltonian density.
In the most relevant cases $h(x)$ is a two site operator, but we do
not necessarily need this restriction. 
For simplicity we require periodic boundary conditions.

In integrable models there exists a family of conserved operators $Q_\alpha$,
where $\alpha$ is an index or multi-index. They mutually commute:
\begin{equation}
  [Q_\alpha,Q_\beta]=0,
\end{equation}
and the Hamiltonian is a member of the series. 

We concentrate on the strictly local operators, which
are given as
\begin{equation}
  \label{Qxdef}
  Q_\alpha=\sum_{x=1}^L Q_\alpha(x),
\end{equation}
where  $Q_\alpha(x)$
is a short range operator identified 
as the charge density. It is important that in a finite volume
$Q_\alpha$ is well defined for volumes larger than the range of
$Q_\alpha(x)$, and then the density $Q_\alpha(x)$ does not depend further
on $L$. Details on the canonical construction of the charges, and a
few concrete examples will be
given in \ref{sec:construction}.

The Hamiltonian is a member of the series, thus the global operator $Q_\alpha$
is conserved:
\begin{equation}
  \label{comm1}
\frac{d}{dt} Q_\alpha=i  [H,Q_\alpha]=0.
\end{equation}
We are interested in non-equilibrium processes and a hydrodynamic
description, therefore we investigate the time evolution of the charge contained in a
finite section. This leads to 
a continuity relation in
operator form:
\begin{equation}
  \label{Jdef}
\frac{d}{dt}  \sum_{x=x_1}^{x_2}  Q_\alpha(x)=
i  \left[H,\sum_{x=x_1}^{x_2}  Q_\alpha(x)\right]=J_\alpha(x_1)-J_\alpha(x_2+1).
\end{equation}
This relation defines the current operator $J_\alpha(x)$ associated to $Q_\alpha(x)$ under the time
evolution of $H$. The existence of $J_\alpha(x)$ follows simply from
\eqref{comm1} and locality arguments.

The relations \eqref{Qxdef} and \eqref{Jdef} do not
define the $Q_\alpha(x)$ and $J_\alpha(x)$ operators uniquely; certain subtleties
are discussed in Section \ref{sec:construction}. We just put forward
that the additive normalization of the current operators is chosen by
the physical requirement that
\begin{equation}
  \bra{0}J_\alpha(x)\ket{0}=0,
\end{equation}
where $\ket{0}$ is the vacuum or reference state with no particles.

Our goal is to determine the mean values
\begin{equation}
  \bra{n}J_\alpha(x)\ket{n},
\end{equation}
where $\ket{n}$ is an arbitrary excited state of the finite volume
Hamiltonian. In the thermodynamic limit these mean values will enter
the flow equations \eqref{eqstate1}.

In the models in question the exact eigenstates are found using the
Bethe Ansatz \cite{Bethe-XXX,Korepin-Book}. The states are
characterized by a set of lattice momenta
$\{p_1,\dots,p_N\}$ that describes the interacting spin waves.

The un-normalized Bethe wave function can be written for $x_1<x_2<\dots<x_N$ as \cite{Bethe-XXX}
\begin{equation}
  \label{psibethe}
  \begin{split}
&  \Psi(x_1,x_2,\dots,x_N)=\\
& \hspace{0.5cm} \sum_{\sigma\in S_N} \left[\exp(i\sum_{j=1}^N p_{\sigma_j}x_j) 
    \mathop{\prod_{j<k}}_{\sigma_j>\sigma_k} S(p_j,p_k)\right].
    \end{split}
\end{equation}
Here each term in the sum represents free wave propagation with a
given spatial ordering of the particles. The amplitude
$S(p_j,p_k)=e^{i\delta(p_j,p_k) }$ is a relative phase between
terms with different particle ordering, and it can be interpreted as
the two-particle scattering amplitude. It depends on the model in
question, and can be determined from the two-particle problem.
The wave function is two-particle
reducible: any multi-particle interaction is explicitly factorized into a
succession of two-particle scatterings.

These wave functions describe eigenstates if they are periodic, from
which we obtain the Bethe equations:
\begin{equation}
  \label{bethe0}
  e^{ip_jL}\prod_{k\ne j} S(p_j,p_k)=1,\qquad
  j=1\dots N.
\end{equation}

It is useful to introduce the rapidity parametrization $p=p(\lambda)$,
where $\lambda$ is the
additive parameter for the scattering phase:
\begin{equation}
  S(p_j,p_k)=S(p(\lambda_j),p(\lambda_k))=S(\lambda_j-\lambda_k).
\end{equation}
The Bethe equations can then be written as
\begin{equation}
  \label{bethe}
  e^{ip(\lambda_j)L}\prod_{k\ne j} S(\lambda_j-\lambda_k)=1,\qquad
  j=1\dots N.
\end{equation}
In the following the normalized $N$-particle Bethe states with
rapidities $\{\lambda\}_N=\{\lambda_1,\dots,\lambda_N\}$ will be denoted as
$\ket{\{\lambda\}_N}$.

The total energy and lattice momentum can be computed additively:
\begin{equation}
  \begin{split}
  \label{Edef}
  E&=\sum_{j=1}^N e(\lambda_j)\\
 P&=\sum_{j=1}^N p(\lambda_j)\quad \text{mod } 2\pi,
  \end{split}
\end{equation}
where the single-particle energy $e(\lambda)$ is a further characteristic function of the model.

Similarly, the eigenvalues of the conserved charges
are
\begin{equation}
  \label{Qmean}
  Q_\alpha\ket{\lambda_1,\dots,\lambda_N}
  =\left[\sum_{j=1}^N q_\alpha(\lambda_j) \right] \ket{\lambda_1,\dots,\lambda_N},
\end{equation}
where $q_\alpha(\lambda)$ are the one-particle eigenvalues.

For later use let us write the Bethe equations in the logarithmic form:
\begin{equation}
\label{Betheeq}
  p(\lambda_j)L+\sum_{k\ne j} \delta(\lambda_j-\lambda_k)=2\pi  I_j,
  \quad  j=1\dots N,
\end{equation}
Here $I_j\in \mathbb{Z}$ are the momentum quantum numbers, which
can be used to parametrize the states.

In our derivation an important role will be played by the so-called
Gaudin matrix $G$, which is defined as
\begin{equation}
  \label{gaudin2}
  G_{jk}=\frac{\partial}{\partial \lambda_k} (2\pi I_j),\qquad
  j,k=1\dots N,
\end{equation}
where now the $I_j$ are regarded
as functions of the rapidities.
Explicitly we have
\begin{equation}
  \label{Gaudin}
  G_{jk}=\delta_{jk}\left[p'(\lambda_j)L+\sum_{l=1}^N \varphi(\lambda_j-\lambda_l)
  \right]-\varphi(\lambda_j-\lambda_k),
\end{equation}
where
\begin{equation}
  \label{phidef}
  \varphi(\lambda)=-i\frac{\partial}{\partial \lambda} \log(S(\lambda)).
\end{equation}
There are two interpretations of the Gaudin matrix. First, $\det G$
describes the density of states in rapidity
space. This follows from the fact that in the space of the quantum
numbers the states are evenly distributed, and $G$ is defined as the
Jacobian of the mapping from $\lambda_j$ to $I_j$. Second, the Gaudin
determinant also describes the norm of the Bethe Ansatz wave function in many
integrable models (See Section \ref{sec:construction} and \cite{XXZ-gaudin-norms,korepin-norms}).

\subsection{Main result}

Our main result for the normalized mean values of the current operators is the following:
\begin{equation}
  \label{main}
  \bra{\{\lambda\}_N}J_\alpha(x)  \ket{\{\lambda\}_N}=
{\bf e'} \cdot G^{-1} \cdot {\bf q}_\alpha.
\end{equation}
Here the quantities ${\bf e'}$
and ${\bf q}_\alpha$ are $N$-dimensional vectors with elements
\begin{equation}
  ({\bf e'})_j=\frac{\partial e(\lambda_j)}{\partial \lambda},\qquad
  ({\bf q}_{\alpha})_{j}=q_\alpha(\lambda_j),
\end{equation}
and $G^{-1}$ is the inverse of the Gaudin matrix.

In the simplest case of $N=1$ the Gaudin matrix has a single element
$G_{11}=Lp'(\lambda)$ and \eqref{main} gives the anticipated
classical result
\begin{equation}
   \bra{\lambda}J_\alpha(x)  \ket{\lambda}=\frac{e'(\lambda)q_\alpha(\lambda)}{Lp'(\lambda)}
  =\frac{v(\lambda)q_\alpha(\lambda)}{L},
\end{equation}
where we introduced the bare group velocity
$v(\lambda)=e'(\lambda)/p'(\lambda)=\partial e/\partial p$. 

A similar semi-classical interpretation can be given also for higher particle numbers. 
Using the definition \eqref{gaudin2} and the additive
formula \eqref{Edef}
we find the alternative expression
\begin{equation}
  \label{maskepp}
  \bra{\{\lambda\}_N}J(x)  \ket{\{\lambda\}_N}=
 \frac{1}{L} \sum_{j=1}^N
 \veff(\lambda_j) q_\alpha(\lambda_j),
\end{equation}
where we defined the quantities
\begin{equation}
  \label{veff1}
  \veff(\lambda_j) =  \frac{L}{2\pi} \frac{\partial E}{\partial I_j}.
\end{equation}
In Section \ref{sec:semi} it is explained,
that the $\veff$ can be understood
in a simple semi-classical picture
as
effective velocities, describing the propagation of the individual
bare particles in the presence of the others. 

It is remarkable, that the exact result and the
functional form of the effective velocity are so simple
in the finite volume situation. In the thermodynamic limit the papers
\cite{doyon-ghd,jacopo-ghd} conjectured the formula \eqref{Jtdl}
with the effective speed given by
\begin{equation}
  \label{veffregi}
  \veff(\lambda)=\frac{\partial \eps(\lambda)}{\partial P(\lambda)},
\end{equation}
where $\eps(\lambda)$  and $P(\lambda)$ are the so-called ``dressed energy'' and
``dressed momentum''. These are computed
as the energy and momentum differences as 
we add a particle with rapidity $\lambda$ into a sea of particles. The
``dressing'' takes into account the backflow of the other particles,
which can be computed from
the Bethe equations \eqref{Betheeq}.

The correspondence between \eqref{veffregi} and our \eqref{veff1} is
evident: small changes in the dressed momentum and dressed energy can be
traced back to small changes in the momentum quantum numbers and the
overall finite volume energy, respectively:
\begin{equation}
  \delta \eps(\lambda_j)\sim \delta E,\qquad
  \delta P(\lambda_j)\sim \delta\left(\frac{2\pi I_j}{L}\right).
\end{equation}
This implies that \eqref{veffregi} is indeed the thermodynamic limit
of our \eqref{veff1}. Furthermore, our \eqref{maskepp} can be seen as the
finite volume origin of the thermodynamic formula 
\eqref{Jtdl}.

Eq. \eqref{main} is exact in those
cases when the Bethe wave function is exact; this holds for integrable
spin chains or the 1D Bose gas. On the other hand, in integrable QFT (iQFT)
in finite volume the Bethe wave function is only an
approximation, and in iQFT \eqref{main}
holds up to exponentially small corrections in the volume.

Depending on the model, the bare particles of the Bethe Ansatz can
form bound states. These are described by the so-called string
solutions of the Bethe equations \cite{Takahashi-Book}. It is
important that our formula \eqref{main} is exact on the level of the
individual Bethe rapidities, even in the presence of strings. An
effective description involving the string centers (describing the
rapidities of the composite particles) can be given afterwards using
well established methods \cite{kirillov-korepin-norms-strings}.

The main result \eqref{main} concerns the physical current
operators, that describe the flow under time evolution by the physical
Hamiltonian. 
However, it is also useful to consider certain generalized
current operators, that describe the flow of a given charge under time
evolution generated by some other charge.

Let us therefore consider
two local charges $Q_\alpha$ and $Q_\beta$ belonging to the same
integrable hierarchy, implying that all three
operators $H,Q_\alpha,Q_\beta$ commute with each other.
We define $J_\alpha^\beta$ to be the current of the charge $Q_\alpha$ under
unitary time evolution dictated by $Q_\beta$:
\begin{equation}
  \label{Jab}
i  \left[Q_\beta,\sum_{x=x_1}^{x_2}  Q_\alpha(x)\right]=J_\alpha^\beta(x_1)-J_\alpha^\beta(x_2+1).
\end{equation}
Locality of the charge densities $Q_{\alpha,\beta}(x)$ and the global
relation $[Q_\alpha,Q_\beta]=0$ implies that the operator equation
\eqref{Jab} can always be solved with some short-range $J_\alpha^\beta(x)$.

For the mean values of these generalized current operators we have the
following result:
\begin{equation}
  \label{mostgeneral}
  \bra{\{\lambda\}_N}J_\alpha^\beta(x)  \ket{\{\lambda\}_N}=
{\bf q}'_\beta \cdot G^{-1} \cdot {\bf q}_\alpha.
\end{equation}
Here ${\bf q}'_\beta$ is an $N$-element vector with components
$q_\beta'(\lambda_j)$, where $q_\beta(\lambda)$ is the one-particle
eigenvalue of the charge $Q_\beta$ and the prime denotes
differentiation. The analogy between \eqref{mostgeneral} and
\eqref{main} is evident: only the one-particle eigenvalues of the time
evolution operator are replaced.
A special case and certain symmetry properties of this general
statement are treated in Appendices \ref{sec:special}-\ref{sec:xxxex}. 

In the following Section we describe a semi-classical interpretation of
these results, whereas the full quantum mechanical proof is provided in
Section \ref{sec:proof}.

\section{The semi-classical interpretation}

\label{sec:semi}

Here we present a semi-classical computation which also gives a simple
physical interpretation for the main result \eqref{main}. Our
arguments are very similar to those presented in \cite{flea-gas}, with
the main difference being that we consider finite systems: we are
looking at
the motion of $N$ particles on a finite ring of volume
$L$ (see Fig. \ref{fig:semi}). For convenience we consider here continuum
models and a strictly point-like interaction.

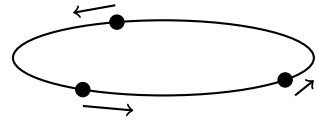
\begin{figure}
  \centering

  \begin{tikzpicture}
    \draw [thick] (0,0) ellipse (2 and 0.5);
 \draw [fill=black] ( -6.180339887e-01,4.755282581e-01) circle (0.7ex);
\draw [thick,->] ( -0.64,0.7) --  (-1.2,0.6);
      \draw [fill=black] (  1.618033989e+00 , -2.938926261e-01) circle  (0.7ex);
\draw [thick,->] (1.75,-0.5) -- (2,-0.3);
 \draw [fill=black] (  -1.071653590e+00,  -4.221639628e-01 )  circle (0.7ex);
  \draw [thick,->] (-1.07,-0.64) -- (-0.4,-0.7);
  \end{tikzpicture}
  
  \caption{The semi-classical picture for the effective
    velocities. There are $N$ particles moving on a circle of
    circumference $L$. The only effect of the interaction are the time
  delays after the scattering events. These accumulate and lead
  to well defined effective speeds in the long time limit. }
  \label{fig:semi}
\end{figure}

Instead of solving the time-dependent Schr\"odinger equation we
employ a semi-classical picture, namely we assume that the particles
can be assigned well-defined straight orbits as long as they don't interact
with each other. In this picture a particle with rapidity $\lambda$
is represented by a wave packet, which travels with the group velocity
\begin{equation}
  v(\lambda)=\frac{de}{dp}=
\frac{e'(\lambda)}{p'(\lambda)}.
\end{equation}
In a typical situation all speeds are different and particles 
meet as they travel around the volume. The scattering events are
taken into account by using exact quantum mechanical solution of the two body
problem. For pure Fourier modes this would mean that for the
scattering of particle $j$ on $k$ the wave function
has to be multiplied by the phase $S(\lambda_j-\lambda_k)$; this is
also reflected by the Bethe Ansatz wave function \eqref{psibethe}.
This momentum
dependent phase results in displacements of the center of the wave
packets
\cite{eisenbud-time-delay,wigner-time-delay,caux-soliton-scattering-time-delay}.
Such time delays are also present in classical integrable models, including
models supporting solitons \cite{classical-integrability-book} or the hard rod gas \cite{hard-rod-gas}.

After a displacement process the particles continue their path with
their own bare speeds, until a further scattering event occurs. The
time delays suffered in each of these events accumulate, and this
alters the actual propagation speed of the wave packets, leading to
the emergence of effective velocities $\veff(\lambda_j)$.
It is
important that due to the higher conservation laws the multi-particle
scattering events always factorize into a succession of two-particle
scatterings \cite{Mussardo:1992uc}, which also implies that the
particle rapidities never change during time evolution.

In this semi-classical picture the current mean values are
evaluated simply as
\begin{equation}
  \label{Jcl1}
  J_{\alpha,cl}=\frac{1}{L}\sum_{j=1}^N \veff(\lambda_j) q_\alpha(\lambda_j).
\end{equation}
Our goal is to find the emerging effective velocities.

We consider long times such that each
pair of particles has scattered on each other many times. In this limit
the particular order of the individual scattering events (which
depends on the initial positions of the particles) becomes
irrelevant.

The spatial displacement of particle $j$ caused by the scattering on the
particle $k$ ($j<k$) is given by the derivative \cite{eisenbud-time-delay,wigner-time-delay,caux-soliton-scattering-time-delay}
\begin{equation}
  \Delta s_{jk} = \frac{\partial \delta(p_j, p_k)}{\partial p_j}=
\frac{\varphi(\lambda_j-\lambda_k)}{p'(\lambda_j)},
  \label{eq:displacement1}
\end{equation}
where we used the rapidity parametrization and
 $\varphi(u)$ is defined in \eqref{phidef}.
This formula is valid when particle $j$ overtakes particle $k$ from the
left, i.e. when $v_j>v_k$. We can formally extend it as
\begin{equation}
  \label{eq:d2}
   \Delta s_{jk} =
\sigma_{jk}
   \frac{\varphi(\lambda_j-\lambda_k)}{p'(\lambda_j)},
\end{equation}
where
\begin{equation}
  \sigma_{jk}=
  \begin{cases}
    +1 & \text{ if } \veff(\lambda_j)>\veff(\lambda_k)\\
    -1 & \text{ if } \veff(\lambda_j)<\veff(\lambda_k).\\
  \end{cases}
\end{equation}

The time elapsed between the two scattering events of the same two particles $j$ and
$k$ can be expressed as: 
\begin{equation}
    T_{jk} = \frac{L}{\abs{\veff(\lambda_j)-\veff(\lambda_k)}}.
\end{equation}
For asymptotically long times $t$ the
accumulated displacement that particle $j$ suffers is given by:
\begin{equation}
    \Delta s_j(t) = \sum_{{\substack{k=1\\k\neq j}}}^N \Delta s_{jk}\cdot \frac{t}{T_{jk}}.
\end{equation}
This causes the difference between the effective and bare velocities:
\begin{equation}
\Delta s_j(t) = \Big(v(\lambda_j)-\veff(\lambda_j)\Big)t.
\end{equation}
Putting everything together we obtain the self-consistent relation
\begin{equation}
  v(\lambda_j) = \veff(\lambda_j) +
  \frac{1}{L} \sum_{{\substack{k=1\\k\neq j}}}^N \Delta s_{jk} \abs{\veff(\lambda_j)-\veff(\lambda_k)}.
    \label{eq:velocities}
  \end{equation}
As an effect of the extra sign $\sigma_{jk}$ in \eqref{eq:d2} it can
be written as
\begin{equation}
  v(\lambda_j) = \veff(\lambda_j)
  + \frac{1}{L} \sum_{{\substack{k=1\\k\neq j}}}^N \frac{\varphi_{jk}}{p'(\lambda_j)} (\veff(\lambda_j)-\veff(\lambda_k)).
    \label{eq:velocities2}
  \end{equation}
Multiplying with $Lp'(\lambda_j)$ and using the definition of the bare
group velocity we get
\begin{equation}
  e'(\lambda_j) L= Lp'(\lambda_j)\veff(\lambda_j)
  + \sum_{{\substack{k=1\\k\neq j}}}^N \varphi_{jk} (\veff(\lambda_j)-\veff(\lambda_k)).
    \label{eq:velocities3}
  \end{equation}
On the r.h.s. we can recognize the action of the Gaudin matrix
\eqref{Gaudin}, we can thus write
  \begin{equation}
L    {\bf e}' = G \cdot {\bf v}_{\text{eff}}.
\end{equation}
Multiplying with $G^{-1}$, and substituting the $\veff(\lambda_j)$
into \eqref{Jcl1} leads to
\begin{equation}
  \label{Jcl2}
  J_{\alpha,cl}=
  {\bf e}' \cdot G^{-1} \cdot {\bf q}_\alpha,
\end{equation}
which is identical to the full quantum mechanical result
\eqref{main}. We have thus demonstrated a complete quantum/classical
correspondence with a finite number of particles.

In our derivation we assumed that
$\veff(\lambda_j)>\veff(\lambda_k)$ whenever
$v(\lambda_j)>v(\lambda_k)$. However natural this requirement may
seem, there can be situations where it does
not hold \cite{marci-flea}. In those cases the quantum result remains
valid, but the semi-classical picture can not be applied.

\section{Proof using a form factor expansion}

\label{sec:proof}

Here we present a model-independent proof of our main result for the
current mean values. Our technique relies on a  finite volume
form factor expansion theorem. The proof of this expansion theorem can
depend on the particular model, but the computations of this Section
are quite general.

In the most part we will concentrate only on the physical currents
$J_\alpha$. The generalized currents defined in \eqref{Jab} will be
considered in \ref{sec:Jab}.

The starting point is to use the definition of the current operators \eqref{Jdef},
and to consider matrix elements of this operator relation. We intend
to compute the mean values of the currents, but taking the mean values
of \eqref{Jdef} gives automatically zero on both sides, due to the
Bethe states being translationally invariant and eigenstates of
$H$. Instead, it is useful to take the off-diagonal matrix elements
between two Bethe states with non-equal total lattice momentum:
\begin{equation}
  \label{QJ}
  \begin{split}
 &  i \left(\sum_{j=1}^N e(\lambda_j)-e(\mu_j)\right)
   \bra{\{\lambda\}_N}Q_\alpha(x)\ket{\{\mu\}_N}=\\
&   = \left(1-\prod_{j=1}^N e^{i(p(\mu_j)-p(\lambda_j))}\right) \bra{\{\lambda\}_N}J_\alpha(x)\ket{\{\mu\}_N}.
  \end{split}
 \end{equation}

In integrable models generic off-diagonal finite volume matrix elements of local operators can be expressed as
\begin{equation}
  \label{fftcsa1}
  \bra{\{\lambda\}_N}\ordo(0)\ket{\{\mu\}_M}=
  \frac{F^\ordo(\{\lambda\}_N|\{\mu\}_M)}{\det G_\mu \det G_\lambda},
\end{equation}
where the function $F^\ordo(\{\lambda\}_N|\{\mu\}_M)$ is the so-called
form factor
and $\det G_\lambda$ and
$\det G_\mu$ describe the norm of the Bethe Ansatz wave function, or
alternatively the density of Bethe states in rapidity space. They are $N\times N$ and $M\times M$ Gaudin
determinants computed from the sets of rapidities $\{\lambda\}_N$ and
$\{\mu\}_M$.

The form factors describe the transition matrix element for the
un-normalized Bethe wave functions \eqref{psi}.
They are meromorphic
functions and they are
completely independent from the volume. The volume dependence of the
physical matrix elements only comes through the solution of the Bethe
equations and the normalization factors.
The properties of the form factors have been investigated both in QFT
\cite{Smirnov-Book} and using Algebraic Bethe Ansatz (ABA)
\cite{Korepin-Book}. A derivation of the analytic properties in the
Lieb-Liniger model was also given using coordinate Bethe Ansatz in
\cite{sajat-nested}. The statement \eqref{fftcsa1} is well known in the literature dealing
with integrable lattice models \cite{Korepin-Book}, and for integrable
QFT it was first written down in \cite{fftcsa1}.

As opposed to the transition matrix elements the mean values of local
operators in Bethe
states can not be expressed directly using the infinite volume form
factors. The reason is the appearance of the so-called disconnected
terms: on a mathematical level they arise from the kinematical poles of the form
factors, whereas their physical interpretation is that they describe
processes when a subset of the particles does not interact with the
local operator. On the other hand, those
processes when some of the particles do interact with the operator are described by certain
diagonal limits of the form factor functions.

In order to describe the finite volume mean values, let us define the
so-called symmetric evaluation of the diagonal form factors as
\begin{equation}
  \begin{split}
    &  F^{\ordo}_s(\{\lambda\}_N)=
    \lim_{\eps\to 0}
 F^\ordo(\lambda_1+\eps,\dots,\lambda_N+\eps|\lambda_N,\dots,\lambda_1),
  \end{split}  
\end{equation}
There is an other often used diagonal limit, called the connected form
factors, but they will not be used in this work and for a thorough
discussion we refer to \cite{fftcsa2}.

It is useful to define the functions
$\rho_N(\lambda_1,\dots,\lambda_N)$ as the $N\times N$ Gaudin determinants
evaluated at the set of rapidities $\{\lambda_1,\dots,\lambda_N\}$. In
the notations we suppress the index $N$ and write simply
\begin{equation}
\rho(\{\lambda\}_N)=\det G_\lambda
\end{equation}
We remind that the Gaudin determinants describe the norms of Bethe
wave functions for eigenstates, ie. for sets of rapidities satisfying
the Bethe equations. On the other hand, the functions
$\rho(\{\lambda\}_N)$ are defined for arbitrary sets of rapidities.

It is useful to write down the first two $\rho(\{\lambda\}_N)$ functions. For $N=1$ we have simply
\begin{equation}
  \label{rho1}
\rho(\lambda)=Lp'(\lambda).
\end{equation}
For $N=2$ the Gaudin matrix is
\begin{equation}
  \label{G2}
  G=   \begin{pmatrix}
    p'(\lambda_1)+\varphi_{12} &  -\varphi_{12} \\
      -\varphi_{12} &  p'(\lambda_2)+\varphi_{12}
    \end{pmatrix},
\end{equation}
and its determinant is
\begin{equation}
  \label{rho2}
  \rho(\lambda_1,\lambda_2)
  =L^2p'(\lambda_1)p'(\lambda_2)+L(p'(\lambda_1)+p'(\lambda_2))\varphi_{12},
\end{equation}
where $\varphi_{12}=\varphi(\lambda_1-\lambda_2)$.

Our proof for the current mean values will be based on the following
expansion theorem.

\begin{thm}
  \label{fftcsathm}
The finite volume mean values of local operators can be computed
through the expansion
\begin{equation}
  \label{fftcsa2}
  \bra{\{\lambda\}_N}\ordo(0)\ket{\{\lambda\}_N}=
\frac{\sum\limits_{\{\lambda^+\}\cup\{\lambda^-\}}
  F^{\ordo}_s(\{\lambda^+\}) \rho(\{\lambda^-\}) }
{\rho(\{\lambda\})},
\end{equation}
where the summation runs over all partitionings of the set of the
rapidities into $\{\lambda^+\}\cup\{\lambda^-\}$.
The partitionings include those cases when either
subset is the empty set, and in these cases it is understood that $\rho(\emptyset)=1$ and
 $F^{\ordo}_s(\emptyset)=\vev{\ordo}$ is the v.e.v. The relation \eqref{fftcsa2}
is exact when the Bethe Ansatz wave functions are exact eigenstates of
the model.
\end{thm}
This theorem was first formulated in \cite{fftcsa2} for integrable
QFT, where the Bethe Ansatz for the finite volume eigenstates is not
exact due to the presence of virtual particles. Therefore, in iQFT the
theorem holds up to corrections exponentially small in the volume. On
the other hand, it is an exact relation in
non-relativistic models including the 1D Bose gas and the Heisenberg
spin chains.
We believe that the theorem is new in the case of the XXZ and XXX spin
chains, and it will be proven rigorously in Section \ref{sec:fftcsa2}. Nevertheless,
similar theorems had been known for particular cases in the ABA literature 
\cite{Korepin-Book}.

In \cite{sajat-LM} it was shown that the LeClair-Mussardo (LM) formula (an
integral series developed for thermal mean values) can be considered a thermodynamic
limit of this expansion theorem, whereas in \cite{sajat-ujLM} it was
shown that alternatively  \eqref{fftcsa2} can be derived from the LM formula using
certain analytic continuations. In \cite{bajnok-diagonal} it was also
shown that in iQFT it can be derived directly using the off-diagonal
relation \eqref{fftcsa1}. Regarding the continuum gas models, the
expansion was proven for certain local operators in the Lieb-Liniger
models in \cite{sajat-XXZ-to-LL}.

We note that there is
an alternative expansion theorem using the connected form factors \cite{fftcsa2}, but it will
not be used here.

The continuity relations yield a connection between
the symmetric form
factors of the charge and current operators. Introducing the short-hand notations
\begin{equation}
  \begin{split}
      \FQa(\{\lambda\}_N)&\equiv F_{s}^{Q_\alpha(0)}(\{\lambda\}_N)\\
  \FJa(\{\lambda\}_N)&\equiv F_{s}^{J_\alpha(0)}(\{\lambda\}_N)
    \end{split}
\end{equation}
we get from  \eqref{QJ}
\begin{equation}
  \label{QJsym}
  \left(\sum_{j=1}^N e'(\lambda_j)\right)
  \FQa(\{\lambda\}_N)=
\left(\sum_{j=1}^N p'(\lambda_j)\right) 
\FJa(\{\lambda\}_N).
  \end{equation}
  
  Our strategy is the following: First we find the symmetric form
  factors of the charge densities by comparing the formula of the expansion theorem
  \ref{fftcsathm} to the 
 known mean values \eqref{Qmean}. Next we use the above
  relation to find the $\FJa$. Finally we use Theorem \ref{fftcsathm} for the current operators: we sum up the resulting
  expansion to obtain \eqref{main}. 
Essentially the same strategy has already been applied in
\cite{doyon-ghd,dinh-long-takato-ghd} directly in the thermodynamic
limit, using the LeClair-Mussardo series. The novelty of our approach is
that we perform these steps in finite volume, and that we also provide
the proof of our expansion theorem for the XXX and XXZ spin chains (see Section \ref{sec:fftcsa2}).

\subsection{The form factors of the charge densities}

We consider \eqref{fftcsa2} in the case of the charge density operator
$Q_\alpha(0)$ and write it as
\begin{equation}
  \label{Qexp}
  \begin{split}
&\frac{\rho(\{\lambda\})}{L}\sum_{j=1}^N q_\alpha(\lambda_j) =
   \sum\limits_{\{\lambda^+\}\cup\{\lambda^-\}} \FQa(\{\lambda^+\}) \rho(\{\lambda^-\}).
  \end{split}
\end{equation}
These are algebraic relations that hold for any particle number $N$
and any finite volume $L$. The symmetric diagonal form factors can be
extracted using a recursive procedure: we consider the above algebraic
relations for $N=1,2,\dots$, and at each $N$ we compute the
$N$-particle form factor by subtracting the terms
with all $\FQa(\{\lambda^+\})$ with a lower number of
particles, obtained earlier.

At $N=1$ the relation
immediately gives
\begin{equation}
  \label{FQ1}
\FQa(\lambda)=p'(\lambda)q_\alpha(\lambda),
\end{equation}
where we used \eqref{rho1} and substituted $\vev{Q_\alpha(0)}=0$.

At $N=2$ we use \eqref{rho1}-\eqref{rho2}. 
Substituting them into \eqref{Qexp} and using also \eqref{FQ1} we
observe the cancellation of some terms, leading eventually to
\begin{equation}
  \FQa(\lambda_1,\lambda_2)=
 (q_\alpha(\lambda_1)+q_\alpha(\lambda_2)) (p'(\lambda_1)+p'(\lambda_2))\varphi_{12}.
\end{equation}

This procedure does in principle yield all higher symmetric
form factors. However, the direct recursive subtractions for $N\ge 3$ become more
involved and it is advantageous to use an alternative method. For the
computation of the Gaudin determinants we will apply a graph
theoretical matrix-tree
theorem, which has already proven to be
useful for different problems 
\cite{ivan-didina-dinh-long-1,ivan-didina-dinh-long-2,dinh-long-takato-ghd}. 

Let us introduce the following definitions.
Given a graph $\Gamma$ a directed graph $\mathcal{F}$ is a directed spanning forest of $\Gamma$
if it satisfies the following:
\begin{itemize}
\item $\mathcal{F}$ includes all vertices of $\Gamma$.
\item $\mathcal{F}$ does not include any circles.
\item Each vertex has at most one incoming edge.
\end{itemize}
The nodes without incoming edges are called roots. Each spanning
forest can be decomposed as a union of spanning trees, which are the
 connected components of the forest. It can be seen from the
above definitions, that each spanning tree has
exactly one root. 

\begin{thm}
  \label{matrix-tree-thm}
 Let $G$ be an $N\times N$ matrix obtained as the difference $G=D-K$, where
 $D$ is diagonal and $K$ satisfies the property
 \begin{equation}
   \sum_{k=1}^N K_{jk}=0,\quad j=1,\dots,N.
 \end{equation}
In this case the determinant of $G$ can be expressed as a sum over the
directed spanning forests of the complete graph with $N$ nodes. For
each spanning forest $\mathcal{F}$ let $R(\mathcal{F})$ denote the set
of the roots, and let $l_{jk}$ denote the edges of $\mathcal{F}$
pointing from node $j$ to $k$. Then we have
\begin{equation}
  \label{Gexp}
  \det G=\sum_{\mathcal{F}} \prod_{j\in R(\mathcal{F})} D_{jj}
  \prod_{l_{jk}\in \mathcal{F}} K_{jk}.
\end{equation}
\end{thm}
For a proof see for example
\cite{matrix-tree-theorems}.

The Gaudin determinant given by \eqref{Gaudin} satisfies the
requirements of this Theorem with $D_{jj}=p'(\lambda_j)L$ and
$K_{jk}=\varphi(\lambda_j-\lambda_k)$. The main idea
to obtain the form factors from \eqref{Qexp} is to consider the formal
$L\to 0$ limit of this relation. To do this we need to observe the
$L\to 0$ limit of the various Gaudin determinants. Each term in
\eqref{Gexp} carries a factor of $L^r$ where $r\ge 1$ is 
the number of roots for the particular spanning forest
$\mathcal{F}$. On the r.h.s. of \eqref{Qexp} we have
\begin{equation}
  \lim_{L\to 0} \rho(\{\lambda^-\})=0
\end{equation}
for every non-empty set of $\{\lambda^-\}$. Thus the only term on the 
r.h.s. which survives the $L\to 0$ limit is the one where
$\{\lambda_-\}=\emptyset$ yielding
$\FQa(\{\lambda\})$. On the other hand, on the l.h.s. we
need to keep the $\ordo(L)$ term in $\rho(\{\lambda\})$, which according to the
above theorem gives an expansion over all directed spanning trees $\mathcal{F}'$:
\begin{equation}
  \lim_{L\to 0}\frac{\rho(\{\lambda\}_N)}{L}=
  \sum_{\mathcal{F}'} \left[ p'(\lambda_r)
  \prod_{l_{jk}\in \mathcal{F}'} \varphi_{jk}\right],
\end{equation}
where for each spanning tree $\mathcal{F}'$ the index $r$ denotes its
root vertex, and we used the abbreviation $\varphi_{jk}=\varphi(\lambda_j-\lambda_k)$.

Each directed spanning tree can be obtained uniquely from a
non-directed spanning tree by selecting the single root vertex and
choosing the directions of the edges accordingly. In our case the
function $\varphi$ is symmetric, thus the direction of the edges does
not influence the factors of  $\varphi(\lambda_u-\lambda_v)$. Each
vertex has to be chosen exactly one time as a root, thus we obtain
\begin{equation}
  \lim_{L\to 0}\frac{\rho(\{\lambda\}_N)}{L}=
 \left[\sum_{j=1}^N p'(\lambda_j)\right]
  \sum_{\mathcal{T}}  \prod_{l_{jk}\in \mathcal{T}} \varphi_{jk},
\end{equation}
where the summation runs over the non-directed spanning trees $\mathcal{T}$.

Finally
\begin{equation}
  \label{FQs}
  \FQa(\{\lambda\}_N)=
  \left[\sum_{j=1}^N p'(\lambda_j)\right]
  \left[\sum_{j=1}^N q_\alpha(\lambda_j)\right]
  \sum_{\mathcal{T}}  \prod_{l_{jk}\in \mathcal{T}} \varphi_{jk}.
\end{equation}

\subsection{Summation for the current operators}

Eqs. \eqref{FQs} and \eqref{QJsym} yield the symmetric diagonal form
factors of the current operators:
\begin{equation}
  \label{FJs}
  \FJa(\{\lambda\}_N)=
  \left[\sum_{j=1}^N e'(\lambda_j)\right]
    \left[\sum_{j=1}^N q_\alpha(\lambda_j)\right]
  \sum_{\mathcal{T}}  \prod_{l_{jk}\in \mathcal{T}} \varphi_{jk}.
\end{equation}
Our task is to sum up the expansion for the currents
\begin{equation}
  \label{fftcsa2J}
  \bra{\{\lambda\}_N}J_\alpha(0)\ket{\{\lambda\}_N}=
\frac{\sum\limits_{\{\lambda^+\}\cup\{\lambda^-\}} \FJa(\{\lambda^+\}) \rho({\lambda^-})}
{\rho(\{\lambda\})}.
\end{equation}
Once again it is instructive to consider the first few cases.

At $N=1$ we have
\begin{equation}
  \label{FJ1}
  \FJa(\lambda)=
e'(\lambda)q_\alpha(\lambda).
\end{equation}
Using $\vev{J_\alpha(0)}_0=0$ and \eqref{rho1} gives immediately
\begin{equation}
  \bra{\lambda}J_\alpha(0)\ket{\lambda}=
  \frac{1}{L}\frac{e'(\lambda)}{p'(\lambda)}q_\alpha(\lambda),
\end{equation}
as anticipated.

At $N=2$ there are 3 non-zero terms in the summation:
\begin{equation}
  \label{inter1}
  \frac{\FJa(\lambda_1,\lambda_2)
    +\FJa(\lambda_1)\rho(\lambda_2)+\FJa(\lambda_2)\rho(\lambda_1)}
  {\rho(\lambda_1,\lambda_2)}.
\end{equation}

The two-particle symmetric form factor is
\begin{equation}
    \FJa(\lambda_1,\lambda_2)= (e'(\lambda_1)+e'(\lambda_2))
 (q_\alpha(\lambda_1)+q_\alpha(\lambda_2))\varphi_{12}.
\end{equation}
Substituting this and also \eqref{FJ1} and \eqref{rho1} into \eqref{inter1}
we can express the mean value in the product form
\begin{equation}
  \frac{
    \begin{pmatrix}
      e'(\lambda_1) & e'(\lambda_2)
    \end{pmatrix}
    \begin{pmatrix}
    p'(\lambda_2)+\varphi_{12} &  \varphi_{12} \\
      \varphi_{12} &  p'(\lambda_1)+\varphi_{12}
    \end{pmatrix}
    \begin{pmatrix}
      q_\alpha(\lambda_1) \\ q_\alpha(\lambda_2)
    \end{pmatrix}
  }{\rho(\lambda_1,\lambda_2)}.
\end{equation}
We can recognize the inverse of the two-particle
Gaudin matrix \eqref{G2}. Thus in this case we have obtained
\begin{equation}
    \bra{\lambda_1,\lambda_2}J_\alpha(0)\ket{\lambda_1,\lambda_2}=
   \begin{pmatrix}
      e'(\lambda_1) & e'(\lambda_2)
    \end{pmatrix}
G^{-1}
    \begin{pmatrix}
      q_\alpha(\lambda_1) \\ q_\alpha(\lambda_2)
    \end{pmatrix},
  \end{equation}
as stated in our main result \eqref{main}.

The summation for $N\ge 3$ is considerably more involved.
From the structure of the form factors and the Gaudin determinants we
can see  that for each pair $jk$ of particles there will be various
contributions including the factors $e'(\lambda_j)q_\alpha(\lambda_k)$, stemming
from different terms in \eqref{fftcsa2J}. Our main result
\eqref{main} states that the sum of these terms
will reproduce the $jk$ element of the matrix $G^{-1}$. For this
inverse matrix we can use the
formula
\begin{equation}
  G^{-1}=\frac{\text{adj}(G)}{\det G},
\end{equation}
where $\text{adj}(G)$ is the so-called adjungate matrix of $G$. The Gaudin
determinant is present in the denominator of \eqref{fftcsa2J}, thus
\eqref{main} holds if in the nominator 
the sum of the terms with $e'(\lambda_j)q_\alpha(\lambda_k)$ reproduce the
$jk$ components of $\text{adj}(G)$. The proof of this is not trivial,
and it is presented in Appendix \ref{sec:adjproof}.

\subsection{Mean values of the generalized currents}

\label{sec:Jab}

It is straightforward to repeat the previous calculations for the case
of the generalized currents $J^\beta_\alpha$ defined in
\eqref{Jab}. For the symmetric diagonal form factors we will use the
notation $\FJab(\{\lambda\}_N)$.
The local continuity equation then leads to
\begin{equation}
  \label{Jabsym}
  \left(\sum_{j=1}^N q_\beta'(\lambda_j)\right)
  \FQa(\{\lambda\}_N)=
\left(\sum_{j=1}^N p'(\lambda_j)\right) 
\FJab(\{\lambda\}_N),
  \end{equation}
  from which
\begin{equation}
  \label{FJab}
\FJab(\{\lambda\}_N)=
  \left[\sum_{j=1}^N q_\beta'(\lambda_j)\right]
    \left[\sum_{j=1}^N q_\alpha(\lambda_j)\right]
  \sum_{\mathcal{T}}  \prod_{l_{jk}\in \mathcal{T}} \varphi_{jk}.
\end{equation}
These form factors are needed to sum up the expansion theorem
  \eqref{fftcsa2}. All of the previous computations can be applied
  by exchanging the function $e'(\lambda)$ with $q_\beta'(\lambda)$.
 This has a simple interpretation: for the general currents the time
  evolution is dictated by $Q_\beta$ instead of the physical Hamiltonian.
  Performing all the previous steps
  we obtain our result \eqref{mostgeneral}.

\section{Current operators and Algebraic Bethe Ansatz in integrable lattice models}

\label{sec:construction}

Here we review the canonical construction of the conserved charge and
associated current operators; we also sketch the standard method of
the Algebraic Bethe Ansatz to find the eigenstates.
We will concentrate on lattice models that are obtained from
integrable Lax operators. The treatment below is rather general, with
concrete examples given later in \ref{sec:examples1}.

Let $\mathcal{H}=V_1\otimes
V_2\otimes\dots\otimes V_L$ denote the Hilbert space of the model,
with $V_j\approx \complex^D$ with some $D=2,3,\dots$ and let
$V_a\approx \complex^{\tilde D}$ with some $\tilde D$ denote
the so-called auxiliary space. In our main examples $\tilde D=D$, but
this is not necessary. Let $\mathcal{L}(u)$ denote the
so-called Lax operator acting on $V_j\otimes V_a$. Here $u\in\complex$ is the
spectral parameter.

The monodromy matrix $T(u)$ acting on $V_a\otimes \mathcal{H}$ is
defined as
\begin{equation}
  \label{monodromy}
T(u)=\mathcal{L}_{1,a}(u)\dots\mathcal{L}_{L,a}(u).
 \end{equation}
The transfer matrix is given by the trace in auxiliary space, which
corresponds to enforcing periodic boundary conditions:
\begin{equation}
     t(u)=\text{Tr}_a T(u).
\end{equation}
The local Lax operators satisfy the exchange property
\begin{equation}
\mathcal{L}_{1a}(u)\mathcal{L}_{1b}(v)R_{ab}(v-u) =
R_{ab}(v-u)\mathcal{L}_{1b}(v)\mathcal{L}_{1a}(u),
  \label{eq:RLL}
\end{equation}
where we introduced two auxiliary spaces $V_{a,b}$ and $R(u)$ is the
so-called $R$-matrix acting on $V_a\otimes V_b$. It satisfies the
Yang-Baxter relation
\begin{equation}
  \begin{split}
&  R_{12}(u_1-u_2)R_{13}(u_1-u_3)R_{23}(u_2-u_3) =\\
&\hspace{0.5cm}=R_{23}(u_2-u_3)R_{13}(u_1-u_3)R_{12}(u_1-u_2),
  \end{split}
  \label{eq:YB}
\end{equation}
which is a relation of operators acting on the triple product
$V_1\otimes V_2\otimes V_3$ of auxiliary spaces.

As an effect of these relations the transfer matrices form a commuting
family of operators \cite{Korepin-Book}:
\begin{equation}
  [t(u),t(v)]=0.
\end{equation}
The transfer matrix encodes the hierarchy of the conserved charges,
which are obtained by expanding $t(u)$ around certain special points.

We consider models where the dimensions of the physical and
auxiliary spaces are equal and the local Lax operator can be chosen to
be identical to the $R$-matrix: $\mathcal{L}(u)=R(u)$.
Furthermore, we will consider cases where the $R$-matrix satisfies the
initial condition $R(0)=P$ with $P$ being the permutation operator, such that 
the transfer matrix satisfies
\begin{equation}
  t(0)=U,
\end{equation}
where $U$ is the translation operator by one site.

For any $\alpha\in \mathbb{N}$, $\alpha\ge 2$ we then define
\begin{equation}
  \label{Qalfadef}
  Q_\alpha=i\left.\left(\frac{\partial}{\partial u}\right)^{\alpha-1} \log(t(u))\right|_{u=0}.
\end{equation}
These charges will be extensive and they can be
written as a sum over the charge density operators \cite{Luscher-conserved}: 
\begin{equation}
  \label{Qalphadef}
  Q_\alpha=\sum_{j=1}^L Q_\alpha(x).
\end{equation}
With this definition $Q_\alpha(x)$ spans $\alpha$ sites, and for
$\alpha=2$ we have
\begin{equation}
  \label{q2kappa}
  Q_2=\kappa H,
\end{equation}
with $H$ being the physical Hamiltonian and $\kappa$ being an
$L$-independent factor, which depends on the particular conventions
that are used.
If the spectral parameter is chosen appropriately, then the 
canonical $Q_\alpha$ defined above are all Hermitian.

We consider models with particle number conservation. In these cases
there is a reference state $\ket{0}$ with zero particles present. We
require that the overall normalization of the transfer matrix 
satisfies
\begin{equation}
  \bra{0}t(u)\ket{0}=1,
\end{equation}
leading to
\begin{equation}
  \label{addi}
  \bra{0}Q_\alpha\ket{0}=0.
\end{equation}
Additional multiplicative factors in $t(u)$ would alter the additive
normalization of the charges.

Having found the charge densities, the continuity
relations \eqref{Jdef} and \eqref{Jab} uniquely define the current and
generalized current operators $J_\alpha$ and $J_\alpha^\beta$. It
follows from \eqref{q2kappa} that for $\beta=2$ we can identify $J_\alpha^2=\kappa J_\alpha$.

We note that the relations \eqref{Qalfadef}-\eqref{Qalphadef} do not define $Q_\alpha(x)$ unambiguously: for
any choice $Q_\alpha(x)$ we can take a local operator $D(x)$ and
define an alternative density
\begin{equation}
  \label{Qgauge}
  Q_\alpha'(x)=Q_\alpha(x)+D(x+1)-D(x),
\end{equation}
which leads to the same integrated charge. This transformation can be
considered as a  ``gauge freedom'' for the definition of the charge
densities, and it is
discussed in detail in \cite{jacopo-benjamin-bernard--ghd-diffusion-long}.
This ``gauge choice'' does not
alter the charge mean values, but it changes the definition of
the current operators as
\begin{equation}
  J^{\beta'}_\alpha(x)=J^\beta_\alpha(x)-i[Q_\beta,D(x)].
\end{equation}
The additional terms do not affect the mean values of $Q_\alpha(x)$ and $J^\beta_\alpha(x)$.

An alternative  way of constructing the charges is with the help of
the boost operator
\cite{Tetelman,sogo-wadati-boost-xxz,Thacker-boost,GM-higher-conserved-XXZ},
which is defined on the infinite chain as the formal expression
\begin{equation}
  \label{boostdef}
  B=\sum_{x=-\infty}^{\infty} x Q_2(x),
\end{equation}
where the density of the charge $Q_2$ is simply
\begin{equation}
  Q_2(x)=i P_{x,x+1}\left.\frac{dR_{x,x+1}(\lambda)}{d\lambda}\right|_{\lambda=0}.
\end{equation}
A formal manipulation shows that \cite{Thacker-boost,GM-higher-conserved-XXZ}
\begin{equation}
  \frac{dt_\infty(\lambda)}{d\lambda}=i [B,t_\infty(\lambda)]+\text{const.}
\end{equation}
where $t_\infty(\lambda)$ is the transfer matrix of an infinite
chain. It follows that
\begin{equation}
  \label{boostrel}
  Q_{\alpha+1}=i[B,Q_\alpha]+\text{const.}
\end{equation}
The additional constant parts are not fixed by the formal
computations, and need to be adjusted afterwards.  Defining a
recursion as $\tilde Q_{\alpha+1}=i [B,Q_\alpha]$ and using
\eqref{addi} we get the correct normalization
\begin{equation}
  Q_{\alpha+1}= \tilde Q_{\alpha+1}-\bra{0} \tilde Q_{\alpha+1}\ket{0}.
\end{equation}

It is also possible to compute the current operators $J^\beta_\alpha(x)$
using a generalization of the boost operator. A formal application of
\eqref{Jab} gives
\begin{equation}
i [Q_\beta,\sum_x x Q_\alpha(x)]=\sum_x  J^\beta_\alpha(x).
\end{equation}
This relation can be used to obtain $J^\beta_\alpha(x)$, but depending on the
situation the direct application of \eqref{Jab} might be more efficient.

\subsection{The Bethe Ansatz solution}

Let us now focus on the case of $D=\bar D=2$, furthermore we consider
models with $U(1)$ symmetry where the fundamental $R$-matrix is of the
form
\begin{equation}
  \label{Rdef}
  R(u)=
  \begin{pmatrix}
    1 & 0 & 0 & 0\\
    0 & b(u) & c(u) & 0 \\
    0 & c(u) & b(u) & 0 \\
    0 & 0 & 0 & 1
  \end{pmatrix}.
\end{equation}
Specific examples will be given later in \ref{sec:examples1}.

The monodromy matrix defined in
\eqref{monodromy} is usually written  in the block form
\begin{equation}
  T(u)=
  \begin{pmatrix}
    A(u) & B(u) \\
    C(u) & D(u)
  \end{pmatrix},
\end{equation}
where the blocks correspond to the degrees of freedom in auxiliary
space, and $A(u), B(u), C(u), D(u)$ are operators acting on the spin
chain.

It follows from the local relation \eqref{eq:RLL} that the monodromy
matrix satisfies 
\begin{equation}
T_{a}(u)T_{b}(v)R_{ab}(v-u) =
R_{ab}(v-u)T_{b}(v)T_{a}(u),
  \label{eq:RTT}
\end{equation}
where $a,b$ refer to two different auxiliary spaces. Commutation
relations between the $A,B,C,D$ operators can be derived from
\eqref{eq:RTT}; they are listed in Appendix \ref{sec:singproof}.

In our models possessing  $U(1)$ symmetry there exists a reference state $\ket{0}$ which is
annihilated by all $C(u)$ for all $u$. Typically
$\ket{0}$ is chosen as the state with all spins up.
Then the Bethe states can be created in Algebraic Bethe Ansatz as
\begin{equation}
  \label{BBethe}
  \ket{\lambda_1,\dots,\lambda_N}=\prod_{j=1}^N B(\lambda_j-\sigma) \ket{0}.
\end{equation}
Here  $\sigma$ is a constant which is chosen later such that rapidity
parametrization becomes exactly the same as in coordinate Bethe Ansatz.

These states are eigenstates of the spin chain Hamiltonian if the rapidities
satisfy the Bethe equations
\begin{equation}
  \frac{a(\lambda_j-\sigma)}{d(\lambda_j-\sigma)}\prod_{k\ne j}
\frac{f(\lambda_k,\lambda_j)}{f(\lambda_j,\lambda_k)}
  =1,
\end{equation}
where we  introduced the function
\begin{equation}
  \label{fdef}
  f(u,v)=\frac{1}{b(u-v)},
\end{equation}
and  the vacuum eigenvalues $a(u),d(u)$ defined as
\begin{equation}
  A(u)\ket{0}=a(u)\ket{0},\qquad D(u)\ket{0}=d(u)\ket{0}.
\end{equation}
For $R$-matrices of the form \eqref{Rdef} we have
\begin{equation}
  a(u)=1,\qquad d(u)=(b(u))^L.
\end{equation}

Dual eigenstates are created as
\begin{equation}
  \bra{\lambda_1,\dots,\lambda_N}=\bra{0} \prod_{j=1}^N C(\lambda_j-\sigma).
\end{equation}
For on-shell states of the physical chain these are the
adjoints of the states \eqref{BBethe} (for detailed proof see
\cite{Korepin-Book}), but in the more general case including 
certain inhomogeneity parameters (leading to non-Hermitian
Hamiltonians) they are only dual vectors. 

For eigenvectors the norm of the Bethe states is
\begin{equation}
  \label{Bethenorm}
  \begin{split}
    &
\bra{0} \prod_{j=1}^N C(\lambda_j)\prod_{j=1}^N B(\lambda_j) \ket{0}=\\
&\hspace{1cm}=(\tilde\kappa)^N\left(\prod_{j\ne k} f(\lambda_j,\lambda_k)\right)
 \det G_\lambda,
 \end{split}
\end{equation}
where $\tilde\kappa$ is a model-dependent constant. This statement was
proven in \cite{korepin-norms} based on the singularity 
properties of general overlaps.

Finally, the eigenvalues of the transfer matrix are
\begin{equation}
  \label{teigs}
  t(u)=
  a(u)\prod_{j=1}^N f(\lambda_j-\sigma,u)+
   d(u)\prod_{j=1}^N f(u,\lambda_j-\sigma).
\end{equation}

\subsection{Main examples}

\label{sec:examples1}

The $SU(2)$-invariant XXX Heisenberg spin chain is given by the Hamiltonian
\begin{equation}
  H=\sum_{j=1}^L (\sigma^x_j\sigma_{j+1}^x+\sigma^y_j\sigma_{j+1}^y+\sigma^z_j\sigma_{j+1}^z-1),
\end{equation}
where $\sigma^a_j$, $a=x,y,z$ are the Pauli matrices acting on site
$j$. This integrable model can be obtained from the $R$-matrix of the
form \eqref{Rdef}
with
\begin{equation}
  b(u)=\frac{u}{u+i},\qquad c(u)=\frac{i}{u+i},
\end{equation}
the shift parameter is $\sigma=i/2$ and $\tilde\kappa=1$.

The canonical definition \eqref{Qalfadef} of the charges gives $Q_2=H/2$
\cite{Korepin-Book}; the additive normalization is such that
\eqref{addi} is satisfied.

The eigenstates of the model organize themselves into
$SU(2)$-multiplets, the Bethe Ansatz gives the highest weight vectors.
The wave functions are of the form \eqref{psibethe} with the rapidity
parametrization given by
\begin{equation}
  \begin{split}
  e^{i p(\lambda)}&= \frac{\lambda-i/2}{\lambda+i/2}\\    
e^{i\delta(\lambda)}&=\frac{\lambda+i}{\lambda-i}.
\end{split}
\end{equation}
The one particle energy eigenvalue is $e(\lambda)=2q_2(\lambda)$
with
\begin{equation}
  \label{q2}
q_2(\lambda)=-p'(\lambda)=-\frac{1}{\lambda^2+1/4}.
\end{equation}
The rapidity parameters take values in the whole complex plain. The
solutions of the Bethe equations organize themselves into strings,
which describe bound states of spin waves \cite{Takahashi-Book}. These
bound states can be regarded as different particle types in 
the thermodynamic limit. However, we perform our finite volume
analysis on the level of the individual Bethe roots, therefore we do
not treat the string solutions separately.

In this model the canonical charges $Q_\alpha$ defined by \eqref{Qalfadef} can be computed
explicitly \cite{GM-higher-conserved-XXX}. For the sake of
completeness we present the explicit formulas up to $Q_4$ in Appendix
\ref{sec:xxxex}; more explicit results are found in  \cite{GM-higher-conserved-XXX}.

Regarding the one-particle eigenvalues of the charges it follows directly from \eqref{teigs}
that
\begin{equation}
  \label{qalfami}
  q_\alpha(\lambda)=
  \left. \left(
\frac{\partial}{\partial x}
  \right)^{\alpha-2} q_2(\lambda-x)\right|_{x=0}.
\end{equation}
The transfer matrix commutes with the global $SU(2)$ transformations,
therefore the $Q_\alpha$ with $\alpha\ge 2$ are all $SU(2)$-invariant
operators. The global spin operators are additional conserved
quantities, and the traditional choice is to add  $Q_1=S^z$ into the
commuting family. If the vacuum is chosen as the reference state with
all spins up, then the one-particle eigenvalues of $Q_1$ are simply
$q_1(\lambda)=-1$. 

Explicit real space formulas for the current operators of the XXX
model are not available. Based on 
\cite{GM-higher-conserved-XXX} it seems plausible that closed form
results can be computed, but we have not pursued this
direction. Nevertheless we computed the first few currents and
generalized currents, the results are presented in \ref{sec:xxxex}.

Our second example is the XXZ model, which is given by the Hamiltonian
\begin{equation}
  H=\sum_{j=1}^L (\sigma^x_j\sigma_{j+1}^x+\sigma^y_j\sigma_{j+1}^y+\Delta(\sigma^z_j\sigma_{j+1}^z-1)).
\end{equation}
This model can be obtained from an $R$-matrix of the form \eqref{Rdef} with
\begin{equation}
  \label{bcdef}
  b(u)=\frac{\sin(u)}{\sin(u+i\eta)},\qquad c(u)=\frac{\sin(i\eta)}{\sin(u+i\eta)},
\end{equation}
where $\Delta=\cosh(\eta)$. In
this normalization we have $Q_2=H/(2\sinh\eta)$. 
The shift parameter is $i\eta/2$, and $\tilde\kappa=\sinh(\eta)$.

For simplicity we focus on the regime
$\Delta>1$, where the rapidity parameters take values in the strip
$\Re(\lambda)\in [-\pi/2,\pi/2]$ and the momentum and scattering
amplitudes are 
\begin{equation}
  \begin{split}
  e^{i p(\lambda)}&= \frac{\sin(\lambda-i\eta/2)}{\sin(\lambda+i\eta/2)}\\    
e^{i\delta(\lambda)}&=\frac{\sin(\lambda+i\eta)}{\sin(\lambda-i\eta)},
\end{split}
\end{equation}
with the one particle energy being $e(\lambda)=2\sinh(\eta)
q_2(\lambda)$ where now
\begin{equation}
  \label{q2xxz}
  q_2(\lambda)= \frac{\sinh(\eta)}{2(\cos(2\lambda)-\cosh(\eta))}.
\end{equation}
The one-particle eigenvalues of the canonical charges $Q_\alpha$ are
again given by \eqref{qalfami}, with the $q_2(\lambda)$ function
above. The real space representation of the $Q_\alpha$
is treated in detail in \cite{GM-higher-conserved-XXZ}, but we do not
use those results here. The $U(1)$-invariance of the model leads to
the independent conserved charge $Q_1=S^z$ with one-particle
eigenvalues 
$q_1(\lambda)=-1$.

\section{Proof of the expansion theorem -- the XXZ chain}

\label{sec:fftcsa2}

Here we prove the Theorem \ref{fftcsathm} using standard methods of the Algebraic
Bethe 
Ansatz (ABA) \cite{Korepin-Book}. In fact, our proof can be considered
a generalization of the proof given by Korepin for the norms of the
Bethe Ansatz wave functions \cite{korepin-norms}. Similar ideas have
been worked out by one of the authors in the work \cite{sajat-LM},
which considered certain local correlators of the continuum 1D Bose gas.
Here we restrict ourselves to the case of
the XXZ spin chain, with the $R$-matrix given by conventions
\eqref{Rdef}-\eqref{bcdef}. The case of the XXX chain can be treated similarly.

For the proof it is useful to define re-normalized operators
\begin{equation}
  \label{renorm}
  \mathbb{B}(u)=\frac{B(u)}{d(u)},\qquad
  \mathbb{C}(u)=\frac{C(u)}{d(u)}.\qquad
\end{equation}
In order to shorten the notations, in this Section we do not use the rapidity shift $-i\eta/2$ that
appeared earlier in \eqref{BBethe}.

Our aim is to derive the form factor expansion for the normalized mean
values
\begin{equation}
  \frac{\bra{0} \prod_{j=1}^N \mathbb{C}(\lambda_j)\ \ordo\ \prod_{j=1}^N \mathbb{B}(\lambda_j) \ket{0}}
  {\bra{0} \prod_{j=1}^N \mathbb{C}(\lambda_j)\prod_{j=1}^N \mathbb{B}(\lambda_j) \ket{0}},
\end{equation}
where $\ordo$ is any operator of the finite spin chain. Quite
interestingly, for this proof we do not require any locality
property from the operator. In fact, locality of an operator is not even a well defined
concept in a finite chain. Thus we don't impose any restriction on
the range of the operator $\ordo$.

To be precise, let us consider the elementary matrices $E_{ab}$,
$a,b=1,2$ and let $E_{ab}(x)$ stand for operator which acts as
$E_{ab}$ on site $x=1,\dots,L$  and with the identity elsewhere. We
perform the proof for an arbitrary product
\begin{equation}
  \label{Eabx}
\ordo=  \prod_{x=1}^L E_{a_xb_x}(x).
\end{equation}
Each operator of the finite chain is a linear combination of these
products. Short range operators are obtained by taking
traces in some subset of the indices $a_x,b_x$.

In order to compute these mean values we need to embed the operators
\eqref{Eabx} in the Yang-Baxter algebra. This procedure is called the
``quantum inverse scattering problem'' and was solved in
\cite{maillet-inverse-scatt,goehmann-korepin-inverse,maillet-terras-inverse}. 

In the case of the homogeneous chain we have
\begin{equation}
  \label{inversesol}
 \prod_{x=L}^1 E_{a_xb_x}(x)=\prod_{x=L}^1  T_{a_xb_x}(0).
\end{equation}

Our strategy is that we prove the theorem for an arbitrary product
\begin{equation}
\ordo=  X(\mu_L)\dots X(\mu_2) X(\mu_1),
\end{equation}
where $X(\mu)$ may represent one of the four operators $A,B,C,D$
evaluated at spectral parameter $\mu$. Afterwards we take the
$\mu_j\to 0$ limit, and by continuity we obtain the statement for
\eqref{inversesol}. 

Our proof is based on the singularity properties of the matrix
elements of operators. We follow closely the proof of Korepin for the
norms of Bethe states \cite{korepin-norms}; in fact, our proof can be
considered a slight generalization of the methods of Korepin. For an
earlier similar proof see \cite{sajat-LM}.

We introduce the following notation for a
matrix element of the arbitrary operator $\ordo$ between two states
(not necessarily eigenstates) described by the rapidity sets
$\{\lambda^B\}_N$ and $\{\lambda^C\}_N$ 
\begin{equation}
\label{matrix_element}
\begin{split}
M_N^\ordo(&\{\lambda^C\}_N,\{\lambda^B\}_N,\{l^C\}_N,\{l^B\}_N)=\\
&=\bra{0}\prod_{j=1}^N\mathbb{C}(\lambda_j^C) \ordo \prod_{j=1}^N\mathbb{B}(\lambda_j^B)\ket{0}.
\end{split}
\end{equation}
These matrix elements depend on the rapidities and the variables
$l(\lambda)=a(\lambda)/d(\lambda)$, where $a(\lambda)$ and
$d(\lambda)$ are the vacuum expectation values of the operators
$A(\lambda)$ and $D(\lambda)$ respectively \cite{Korepin-Book}. For what follows, it is
important, that we can consider arbitrary functions for $a(\lambda)$
and $d(\lambda)$ and therefore $M_N^\ordo$ is the function of $4N$
independent variables. It follows from the commutativity of the $\mathbb{B}$
and $\mathbb{C}$ operators that $M_N^\ordo$ is invariant with respect
to simultaneous exchanges $\lambda_j^B\leftrightarrow\lambda_k^B$,
$l_j^B\leftrightarrow l_k^B$, and similarly for the rapidities on the
left hand side.

The matrix elements have apparent singularities as two rapidities from
the two sides approach each other. These apparent poles result from
the commutation relations between the  $\mathbb{B}$
and $\mathbb{C}$ operators. We will show that the structure of these
poles completely determines the mean values. For simplicity we will
focus on the singularities as $\lambda_N^C-\lambda_N^B\to 0$; the
other cases follow simply from the permutation symmetry.

\begin{thm}
The matrix elements satisfy the following singularity property: 
\begin{equation}
\label{sing_prop}
\begin{split}
&M_N^\ordo(\{\lambda^C\}_N,\{\lambda^B\}_N,\{l^C\}_N,\{l^B\}_N)\xrightarrow{\lambda_N^C\rightarrow\lambda_N^B}\\
&\rightarrow\frac{i\sinh(\eta)}{\lambda_N^C-\lambda_N^B}\big(l_N^C-l_N^B\big)
\left(\prod_{k=1}^{N-1}f_{kN}^Cf_{kN}^B\right)\times\\
&M_{N-1}^\ordo(\{\lambda^C\}_{N-1},\{\lambda^B\}_{N-1},\{l_{mod}^C\}_{N-1},\{l_{mod}^B\}_{N-1}),
\end{split}
\end{equation}
where in the third line the matrix element is calculated with the
modified vacuum expectation values
\begin{equation}
  \label{admod}
  a_{mod}(\lambda)=a(\lambda)f(\lambda_N,\lambda),\quad 
  d_{mod}(\lambda)=d(\lambda)f(\lambda,\lambda_N).
\end{equation}
\end{thm}
The proof is rather technical and it is presented in
Appendix \ref{sec:singproof}.

In the physical case (when the $l_N^{B,C}$ variables are not independent) this expression gives a finite result
\begin{equation}
\frac{i\sinh(\eta)}{\lambda_N^C-\lambda_N^B}\big(l_N^C-l_N^B\big)\xrightarrow{\lambda_N^C\rightarrow\lambda_N^B}
\sinh(\eta) l(\lambda) z(\lambda),
\end{equation}
where $z(\lambda)=i\partial_\lambda\log l(\lambda)$.

The diagonal evaluation of the matrix element is defined as:
\begin{equation}
\begin{split}
&M_{N,d}^\ordo(\{\lambda\}_N,\{l\}_N,\{z\}_N)=\\
&=\lim_{\lambda_k^C\rightarrow\lambda_k^B}M_N^\ordo(\{\lambda^C\}_N,\{\lambda^B\}_N,\{l^C\}_N,\{l^B\}_N) 
\end{split}
\end{equation}
where the limit is performed for every $k=1,\dots,N$. This quantity depends on $3N$ independent
variables $\{\lambda\}_N$, $\{l\}_N$ and $\{z\}_N$. From
\eqref{sing_prop} it follows that the dependence on $z_N$
is linear, and the proportionality factor is
\begin{equation}
\begin{split}
&\frac{\partial}{\partial z_N}M_{N,d}^\ordo(\{\lambda\}_N,\{l\}_N,\{z\}_N)=\\
&=\sinh(\eta)l(\lambda_N)\left(\prod_{k=1}^{N-1}f_{kN}^Cf_{kN}^B\right)\times\\
&\times M_{N-1,d}^\ordo(\{\lambda\}_{N-1},\{l_{mod}\}_{N-1},\{z_{mod}\}_{N-1}),
\end{split}
\end{equation}
where
\begin{equation}
z_{mod}(\lambda)=z(\lambda)+\varphi(\lambda-\lambda_N).
\end{equation}
To calculate the expectation values in the eigenstates of the system,
we have to take the rapidities to the solutions of the Bethe
equations. The quantity obtained after this does not depend on the
$l$-parameters anymore, since they are given by the Bethe equations: 
\begin{equation}
  \begin{split}
    &  \vev{\ordo}_N(\{\lambda\}_N,\{z\}_N)=\\
    &\hspace{0.3cm}=M_{N,d}^\ordo(\{\lambda\}_N,\{l\}_N,\{z\}_N)|_{\{\lambda\} \textrm{ solves B.E.}}.
  \end{split}
\end{equation}
The dependence on $z_N$ is still linear:
\begin{equation}
\label{expval_derivative}
\begin{split}
  &\frac{\partial \vev{\ordo}_N(\{\lambda\}_N,\{z\}_N)}
  {\partial z_N}=
\left(\prod_{k=1}^{N-1}f_{kN}f_{Nk}\right)
\times\\
&\times\sinh(\eta)\vev{\ordo}_{N-1}(\{\lambda\}_{N-1},\{z_{mod}\}_{N-1}).
\end{split}
\end{equation}
To continue the calculation we need to introduce the form factors
$\mathbb{F}^\ordo_N$ and $\mathbb{F}^\ordo_{N,s}$ (these differ from the previously used
form factors only in an overall normalization) 
\begin{equation}
\begin{split}
&\mathbb{F}^\ordo_N(\{\lambda^C\}_N,\{\lambda^B\}_N)=\\
&=M_N^\ordo(\{\lambda^C\}_N,\{\lambda^B\}_N,\{l^C\}_N,\{l^B\}_N)
|_{\text{B.E.}},
\end{split}
\end{equation}
where it is understood that both sets of rapidities solve the Bethe
equations, i.e. we express the $\{l^C\}_N,\{l^B\}_N$ using the
rapidities. From \eqref{sing_prop} it is obvious that this form factor
satisfies the following recursion relation 
\begin{equation}
\begin{split}
\mathbb{F}^\ordo_N&(\{\lambda^C\}_N,\{\lambda^B\}_N)\xrightarrow{\lambda_N^C\rightarrow\lambda_N^B}\\
\rightarrow&\frac{i\sinh(\eta)}{\lambda_N^C-\lambda_N^B}
\left(\prod_{k=1}^{N-1}f_{Nk}^Cf_{kN}^B-\prod_{k=1}^{N-1}f_{kN}^Cf_{Nk}^B\right)\times\\
&\times\mathbb{F}_{N-1}^\ordo(\{\lambda^C\}_{N-1},\{\lambda^B\}_{N-1}).
\end{split}
\end{equation}
By taking the symmetric, diagonal limit of this quantity, one obtains the symmetric form factor
\begin{equation}
\mathbb{F}_{N,s}^\ordo(\{\lambda\}_N)=\lim_{\epsilon\rightarrow 0}\mathbb{F}_N^\ordo(\{\lambda+\epsilon\}_N,\{\lambda\}_N).
\end{equation}

\begin{thm}
The symmetric form factor of the operator $\ordo$ is equal to its
expectation value in the case when every $z_j$ is zero: 
\begin{equation}
\mathbb{F}^\ordo_{N,s}(\{\lambda\}_N)=\vev{\ordo}_N(\{\lambda\}_N,\{0\}_N).
\end{equation}
\end{thm}

\begin{proof}
  From \eqref{sing_prop} it is obvious that the $z$ dependence of the expectation value arises from the rapidity dependence of the $l(\lambda)$ function. This means that the $z$ independent, irreducible part can be obtained by choosing $l(\lambda_j^C)=l(\lambda_j^B)$, where $\{l^B\}_N$ solves the Bethe equations:
\begin{equation}
\begin{split}
&\vev{\ordo}_N(\{\lambda\}_N,\{0\}_N)=\\
&=\lim_{\epsilon\rightarrow 0}M_N^\ordo(\{\lambda^B+\epsilon\}_N,\{\lambda^B\}_N,\{l^B\}_N,\{l^B\}_N).
\end{split}
\end{equation}
On the other hand the symmetric form factor is by definition:
\begin{equation}
  \mathbb{F}^\ordo_{N,s}(\{\lambda\}_N)
  =\lim_{\epsilon\rightarrow 0}M_N^\ordo(\{\lambda^B+\epsilon\}_N,\{\lambda^B\}_N,\{\tilde{l}^B\}_N,\{l^B\}_N),
\end{equation}
where both $\{l^B\}_N$ and  $\{\tilde{l}^B\}_N$ solve the Bethe
equations. But this means that the elements of $\{l^B\}_N$ and
$\{\tilde{l}^B\}_N$ are the products of the appropriate S matrices, so 
\begin{equation}
\{l^B\}_N=\{\tilde{l}^B\}_N.
\end{equation}
This completes the proof.
\end{proof}

We define the $S$ function for an arbitrary bi-partition of the set $\{\lambda\}_N=\{\lambda^+\}_n\cup\{\lambda^-\}_{N-n}$ in the following way:
\begin{equation}
\begin{split}
& S_N(\{\lambda^+\}_n,\{\lambda^-\}_{N-n},\{z^-\}_{N-n})=\big(\sinh(\eta)\big)^{N-n}\times\\
& \times\left(\prod_{\lambda_j\in\{\lambda^+\}}\prod_{\lambda_k\in\{\lambda^-\}}f_{jk}^{+-}f_{kj}^{-+}\right)\left(\prod_{1\leq j<k\leq n}f_{jk}^{--}f_{kj}^{--}\right)\times\\
&\times\rho(\{\lambda^-\}_{N-n},\{z^-\}_{N-n}),
\end{split}
\end{equation}
where $f_{jk}^{+-}=f(\lambda_j^+,\lambda_k^-)$. This function depends
on $z_N$ only through the Gaudin determinant, and its dependence can
be calculated easily by expanding the determinant with respect to its
$N^{\textrm{th}}$ row or column. Doing so one obtains that 
\begin{widetext}
\begin{equation}
\label{derivative_S}
\begin{split}
\frac{1}{\sinh(\eta)}\frac{\partial}{\partial z_N}S_N&(\{\lambda^+\}_n,\{\lambda^-\}_{N-n},\{z\}_N)=\\
&= \left\{ \begin{array}{ll}
\left(\prod\limits_{k=1}^{N-1}f_{kN}f_{Nk}\right)S_{N-1}(\{\lambda^+\}_n,\{\lambda^-\}_{N-n-1},\{z_{mod}^-\}_{N-n-1})&\lambda_N\in\{\lambda^-\}\\
0& \lambda_N\in\{\lambda^+\}
\end{array}
\right..
\end{split}
\end{equation}
\end{widetext}
It is important to notice that if we take all the $z$-s in the argument of $S$ to zero, we get zero:
\begin{equation}
S_N(\{\lambda^+\}_n,\{\lambda^-\}_{N-n},\{0\}_{N-n})=0.
\end{equation}
This follows from the fact that in this case the Gaudin determinant is
zero, which follows from Theorem \ref{matrix-tree-thm}.

With the help of the relations listed above we are now in the position to prove the following theorem:
\begin{thm}
The un-normalized mean value of an arbitrary operator $\ordo$ in any
eigenstate of the system can be calculated in the following way: 
\begin{equation}
\label{thm_expectval}
\begin{split}
&\vev{\ordo}_N({\{\lambda\}_N})=\\
&=\sum_{\{\lambda^+\}\cup\{\lambda^-\}}\!\!\!
\mathbb{F}_{s}^\ordo(\{\lambda^+\})S_N(\{\lambda^+\},\{\lambda^-\},\{z^-\}),
\end{split}
\end{equation}
where the summation goes over every bi-partition of the set of
rapidities, $\{\lambda\}_N$ and for simplicity we did not denote the
cardinality of the sets $\{\lambda^\pm\}$ separately.
\end{thm}
\begin{proof}

We use induction in the particle number $N$. Let us
consider both sides as the multi-linear function of the $z_j$
variables. Our goal is to show that both sides depend the same way on
every $z_j$ and that their $z$ independent parts are also equal. We
look at the first $N$ for which the matrix element defined by
\eqref{matrix_element} is not zero, and we denote this number by $N_{min}$. In
the case when $N<N_{min}$ both sides of \eqref{thm_expectval} are
zero, therefore the equation is satisfied. If $N=N_{min}$ there is
only one non-zero term in the summation in the r.h.s., namely when
$\{\lambda^+\}=\{\lambda\}_N$. This means that on the r.h.s. there is
only $\mathbb{F}^\ordo_{N,s}(\{\lambda\}_N)$, since
$S(\{\lambda^+\},\emptyset,\emptyset)=1$, therefore it is $z$
independent. The l.h.s. is also independent of $z$, which follows from
\eqref{expval_derivative}. But the $z$ independent parts are equal
according to the previously proved theorem. Now let us suppose that
\eqref{thm_expectval} is satisfied for every $N<M$, and examine the
$N=M$ case! In the r.h.s. only those terms depend on $z_j$ where
$z_j\in\{\lambda^-\}$. By taking the partial derivative of it with
respect to $z_j$, the initial summation is getting modified to a new
one, going over all the bi-partitions of the set
$\{\lambda\}_{N-1}=\{\lambda\}_N\backslash \{\lambda_j\}$ (this
follows from \eqref{derivative_S}). According to the induction
assumption this sum gives exactly
$\vev{\ordo}_{N-1}(\{\lambda\}_{N-1},\{z_{mod}\}_{N-1})$, multiplied by $\sinh(\eta)\prod_{k\neq j}f_{kj}f_{jk} $. But using \eqref{expval_derivative}, this demonstrates that the $z$ dependence of the two sides are equal. To
investigate the part independent of the $z$ variables, we have to
take all of them to zero. In this case on the r.h.s. only
$\mathbb{F}_{N,s}^\ordo(\{\lambda\}_N)$ remains. Since
$\vev{\ordo}_N(\{\lambda\}_N,\{0\}_N)=\mathbb{F}_{N,s}^\ordo(\{\lambda\}_N)$,
the two sides are equal. This completes the proof.
\end{proof}

The previously introduced \eqref{fftcsa2} equation is equivalent to
\eqref{thm_expectval}: one obtains the former after
dividing with the norm of the Bethe state given by \eqref{Bethenorm}. 

\section{Connection to the theory of factorized correlation functions}

\label{sec:factorized}

In the case of the XXZ and XXX spin chains our results for the current
mean values are directly related to certain objects in the theory of
factorized correlation functions. In the following we describe this
connection. 

Factorization of correlation functions concerns the equilibrium mean
values of local operators (for a review, see
\cite{kluemper-goehmann-finiteT-review}). Factorization means that any
multi-point 
correlator can be expressed as sums of products of simple building
blocks, which are derived from the two-site density matrix in an
inhomogeneous spin chain. On a practical level the theory consists of
two parts: the algebraic part and the physical part. The algebraic
part deals with the factorization on the level of the operators, and
is independent of the actual physical situation. On the other hand,
the information about the concrete situation is supplied by the
physical part of the computation.

The theory was worked out first for the cases of thermal equilibrium
in infinite volume and for the ground states in finite volume
\cite{physpart,XXZ-massless-corr-numerics-Goehmann-Kluemper,XXZ-massive-corr-numerics-Goehmann-Kluemper,goehmann-kluemper-finite-chain-correlations}. In
\cite{sajat-corr} a conjecture was formulated for the physical part in
any excited equilibrium state in infinite volume; these formulas have been
used to compute the steady state properties after global
quenches. Finally, it was argued in \cite{sajat-corr2} that the known
results for the finite volume ground states
\cite{goehmann-kluemper-finite-chain-correlations} can be extended naturally
to all finite volume excited states, and this leads to the proof of
the formulas of \cite{sajat-corr} in the thermodynamic limit.

In the following we summarize the main statements in the case of the
XXX chain.

The Bethe vectors are highest weight vectors with respect to the
$SU(2)$ symmetry. Let us consider a Bethe state with $\vev{S_z}=0$, thus $N=L/2$. All
such states are $SU(2)$ singlets.

The theory of factorized correlations states
\cite{goehmann-kluemper-finite-chain-correlations,sajat-corr2} that in
the $SU(2)$  singlet states
any multi-point correlation function can be expressed using only a
single generating function $\Psi(x_1,x_2)$ that depends on the excited
state in question. Defining the coefficients
\begin{equation}
 \Psi_{n,m}=\partial_{x_1}^n \partial_{x_2}^m \Psi(x_1,x_2)|_{x_1,x_2=0}
\end{equation}
it can be shown that all correlators can be expressed as finite
combinations of the quantities $\Psi_{n,m}$.
The algebraic part of
the computation expresses the correlators in terms of $\Psi_{n,m}$,
whereas the physical part supplies their actual values, depending on
the Bethe state in question. 

For example the simplest $z-z$
correlators can be expressed as
\begin{align}
\label{sig13}
    \vev{\sigma_{1}^z\sigma_2^z}=&
\frac{1}{3}(1-\Psi_{0,0})\\
   \vev{\sigma_{1}^z\sigma_3^z}
=&\frac{1}{3}(1-4\Psi_{0,0}+\Psi_{1,1}-\frac{1}{2}\Psi_{2,0})\\
\label{sig14}
\begin{split}
  \vev{\sigma_{1}^z\sigma_4^z}=&
\frac{1}{108} (36  +  288 \Psi_{1, 1}   - 15 \Psi_{2, 2} +   10
\Psi_{3, 1}+\\
& +  \Psi_{2, 0}( - 156 + 12 \Psi_{1, 1}  - 6 \Psi_{2, 0})
\\
&+   2 \Psi_{0, 0} (-162  - 42 \Psi_{1, 1}  +   3 \Psi_{2, 2} - 2 \Psi_{3, 1})+\\
&+
\Psi_{1,0}(84 \Psi_{1, 0} - 12 \Psi_{2, 1}+    4 \Psi_{3, 0})).
\end{split}
\end{align}

It was shown in \cite{sajat-corr2} that for the Bethe state
$\ket{\{\lambda\}_N}$ 
the generating function reads
\begin{equation}
\label{psi}
  \Psi(x_1,x_2)=2 {\bf q}(x_1)\cdot G^{-1}\cdot {\bf q}(x_2),
\end{equation}
where ${\bf q}(x_1)$ is a parameter dependent vector of length $N$
with elements ${\bf q}_j(x)=q(\lambda_j-x)$ with $q(\lambda)=1/(\lambda^2+1/4)$.
Furthermore, $G$ is the Gaudin matrix, which now takes the form
\begin{equation}
\label{GaudinXXX}
  G_{jk}=\delta_{jk}\left(
L \frac{1}{u_j^2+1/4}+\sum_{l=1}^N \varphi(u_{jl})
\right)-\varphi(u_{jk}),
\end{equation}
with
\begin{equation}
  \label{phidefXXX}
\varphi(u)=-\frac{2}{u^2+1}.
\end{equation}
The factor of 2 is included in \eqref{psi} to conform with earlier
conventions, see \cite{sajat-corr2}.

It follows that the $\Psi_{n,m}$ coefficients can be expressed as
\begin{equation}
  \Psi_{n,m}=2 {\bf q}_{n+2}\cdot G^{-1}\cdot {\bf q}_{m+2}.
\end{equation}
The shifts in the indices are due to our conventions, namely that the
first member of the series of the charges 
is called $Q_2$.

Comparing to \eqref{mostgeneral} we see that the current mean values
are
\begin{equation}
  \label{mostgeneralXXX}
  \bra{\{\lambda\}_N}J_\alpha^\beta(x)  \ket{\{\lambda\}_N}=
\frac{1}{2} \Psi_{\beta-1,\alpha-2}.
\end{equation}
This gives a new interpretation for the generalized currents:
they are special operators that are represented by a single
$\Psi_{n,m}$.

In the case of non-singlet states the situation is more involved, and
the correlators also involve the so-called moments (see
\cite{boos-goehmann-kluemper-suzuki-fact7,sajat-corr2}). Nevertheless,
the mean values of $SU(2)$ invariant operators are still described by
the $\Psi_{n,m}$, and this is in accordance with the fact that in the
XXX model the charges $Q_\alpha$ and their currents are also
$SU(2)$-invariant.

In the case of the XXZ model spin-flip invariant local correlations
are described by two generating functions (traditionally denoted as
$\omega(x,y)$ and $\omega'(x,y)$, see
\cite{physpart,XXZ-massless-corr-numerics-Goehmann-Kluemper,XXZ-massive-corr-numerics-Goehmann-Kluemper,goehmann-kluemper-finite-chain-correlations}). It
was shown in \cite{sajat-corr2} that for the finite volume excited
states the function $\omega(x,y)$ has a form analogous to \eqref{psi},
thus the Taylor coefficients of this function describe the generalized
current operators in the XXZ model.

At present we don't have an explanation of the observed coincidences
between the factorization and the current mean values. Understanding
this connection might 
lead to an independent proof of our results, at least in the XXZ and
XXX models.

\section{Conclusions and discussion}

\label{sec:conclusions}

We computed an exact finite volume formula for the
current mean values in Bethe Ansatz solvable systems. The main results
are eqs. \eqref{main} and \eqref{mostgeneral}.
We did not treat the direct
thermodynamic limit of these results, but their functional form
implies that they reproduce the previously conjectured 
infinite volume formulas (see eq. \eqref{maskepp} and the discussion
there). We have thus supplied a rigorous foundation for the
treatment of ballistic transport in Generalized Hydrodynamics.

Perhaps the most interesting finding of this work is that the
semi-classical result for the currents remains exact in the
interacting quantum theory, even with a finite number of particles. Ultimately this boils down to the
two-particle reducibility of the Bethe Ansatz wave function
(see eq. \eqref{psibethe}), which also enables the applicability of
the flea gas model to simulate the dynamics \cite{flea-gas}.

Our rigorous proof is rather
general and it relies on a model independent form factor expansion
(Theorem \ref{fftcsathm}). However, the expansion itself has to be proven
separately and the techniques to be applied might vary.
Here we provided a proof for the case of the Heisenberg spin chains.
Whereas we did not treat the Lieb-Liniger
model (1D Bose gas) here, the expansion theorem can be worked out with
essentially the same techniques (see for example \cite{sajat-LM}), at
least for those local charges and currents which have well-defined
real space representations \cite{korepin-LL-higher}.
For integrable QFT the expansion was proven earlier in
\cite{bajnok-diagonal}.

It would be highly desirable to develop an alternative, more transparent
proof for the current mean values. For special cases this can indeed be achieved,
for example the spin current of the XXZ model can be computed using
form factor perturbation theory \footnote{We developed a short unpublished proof 
 in collaboration with Lorenzo Piroli.}, using similar ideas to those
of \cite{kluemper-spin-current}. Nevertheless a simple and general proof
is not available. In the special case of the XXZ and XXX models the
connection to the factorized correlation functions might lead to a
more transparent derivation.

In this work we restricted ourselves to local charges on the lattice. However, it
is known that there exist quasi-local charges that are essential for
the GGE and thus for GHD
\cite{prosen-enej-quasi-local-review,JS-CGGE}.
For quasi-local charges our results hold asymptotically, and the
finite volume formulas \eqref{main}-\eqref{mostgeneral} receive
exponentially small corrections. These are due to the 
long-range contributions to the operators with exponentially
decreasing amplitudes.

It would be interesting to extend our results to multi-component
models solvable by the nested Bethe Ansatz. In these systems much less
is known about local correlations, and the current operators are good
candidates to obtain exact analytic results. In turn, this would give
a rigorous foundation for the hydrodynamic treatment of these models \cite{jacopo-enej-hubbard}. Also, it would be
interesting to consider models without the $U(1)$ symmetry responsible
for particle conservation. The integrable XYZ spin
chain is such a model, where the canonical family of conserved charges
\cite{GM-higher-conserved-XXZ} does not include the $S_z$
operator. Nevertheless there might exist simple exact results for the
current operators, and thus also for GHD. We plan to return to these
questions in future research.

\begin{acknowledgments}
We are grateful to Bruno Bertini, Benjamin Doyon, Enej Ilievski, M\'arton Kormos,
M\'arton Mesty\'an, Lorenzo Piroli, G\'abor Tak\'acs, Eric Vernier, 
Dinh-Long Vu and Takato Yoshimura for useful discussions.
We are thankful in particular to Bruno Bertini and Lorenzo Piroli for drawing
our attention to this problem.
This research was supported by the BME-Nanotechnology FIKP grant
of EMMI (BME FIKP-NAT), by the National Research Development
and Innovation Office (NKFIH) (K-2016 grant no. 119204 and the KH-17
grant no. 125567), and by the ``Premium'' Postdoctoral 
Program of the Hungarian Academy of Sciences.
  \end{acknowledgments}

\appendix

\section{The energy current operator}

\label{sec:special}

Here we investigate a special case of the main formulas
\eqref{main}-\eqref{mostgeneral}:  we consider the energy current $J_H$, which is itself a
conserved operator in integrable lattice models 
In fact, it is proportional to the canonical charge $Q_3$. For
completeness we re-derive this correspondence, and show that
\eqref{main} reproduces the mean values of $Q_3$.

The continuity equation for the energy flow is
\begin{equation}
  [H,h(x)]=i( J_H(x+1)- J_H(x)),
\end{equation}
where $h(x)$ is the Hamiltonian density. It follows that
\begin{equation}
  \label{JHH}
  J_H(x+1)=-i[h(x+1),h(x)].
\end{equation}

Using the definition of the boost operator \eqref{boostdef} and the proportionality
relation \eqref{q2kappa} we compute the canonical $Q_3$ as
\begin{equation}
  \label{Q3boost}
  \begin{split}
    Q_3&
    =i[B,Q_2]=i\kappa^2 [\sum_x xh(x),\sum_y h(y)]=\\
&= i\kappa^2 \sum_x  [h(x+1),h(x)].
  \end{split}
\end{equation}
Therefore we can identify
\begin{equation}
  \label{q3jh}
  Q_3(x)=-\kappa^2 J_H(x).
\end{equation}
The proportionality factor $-\kappa^2$ originates simply in our conventions.

The one-particle eigenvalues of $Q_3$ are obtained from the transfer
matrix construction: \eqref{qalfami} and \eqref{q2} give
\begin{equation}
  q_3(\lambda)=q_2'(\lambda)=\kappa e'(\lambda).
\end{equation}

For $J_H$  \eqref{main} takes the form
\begin{equation}
    \bra{\{\lambda\}_N}J_H(x)  \ket{\{\lambda\}_N}=
{\bf e'} \cdot G^{-1} \cdot {\bf e}.
\end{equation}
We will show that after re-scaling this formula gives the expected mean
value of $Q_3$. 

Let us take an $N$-dimensional vector ${\bf u}$ whose elements are
equal to 1. It follows from the definition of the Gaudin matrix
\eqref{Gaudin} that 
\begin{equation}
  \label{Gvp}
  G {\bf u}=L{\bf p'}.
\end{equation}
Multiplying with the inverse and using $e(\lambda)=q_2(\lambda)/\kappa=-p'(\lambda)/\kappa$
we get
\begin{equation}
  {\bf e'} \cdot G^{-1} \cdot {\bf e}  =-\frac{1}{\kappa L}{\bf e'}\cdot {\bf u}=
-\frac{1}{\kappa^2 L}  \sum_{j=1}^N q_3(\lambda_j).
\end{equation}
With this we have indeed obtained
\begin{equation}
  \label{Q3JH}
    \bra{\{\lambda\}_N}Q_3(x)  \ket{\{\lambda\}_N}=
-\kappa^2  \bra{\{\lambda\}_N}J_H(x)  \ket{\{\lambda\}_N},
\end{equation}
as expected from \eqref{q3jh}.

Eq. \eqref{Q3JH} is a special case of a more general symmetry relation.
It follows from \eqref{qalfami} and the symmetry of the Gaudin matrix,
that the following mean values are equal:
\begin{equation}
  \label{Jszimm}
  \bra{\{\lambda\}_N}J_\alpha^\beta(x)  \ket{\{\lambda\}_N}=
  \bra{\{\lambda\}_N}J_{\beta+1}^{\alpha-1}(x)  \ket{\{\lambda\}_N}.
\end{equation}
This interesting relation was already observed in GHD in
\cite{benjamin-takato-note-ghd}, but a more direct explanation is not
yet known.

We stress that the relation above does not mean that the operators
$J_\alpha^\beta(x)$ and $J_{\beta+1}^{\alpha-1}(x)$ are
 equal; they might differ in total derivatives.
The concrete form of the current operators
depends on the choice of the charge density representing the global
charge operator, and the gauge freedom \eqref{Qgauge} leaves room for re-definitions.
It would be interesting to see whether there is a gauge choice, which
would guarantee that the above
relation holds on the level of the operators.

\begin{widetext}

\section{Explicit formulas for charges and currents in the XXX chain}

\label{sec:xxxex}

In the XXX Heisenberg chain the canonical $Q_2$ operator 
\eqref{Qalfadef} is
\begin{equation}
   Q_2 = \frac{1}{2} \sum_{x=1}^L \left(\underline{\sigma}_x \cdot \underline{\sigma}_{x+1}-1 \right).
\end{equation}
Here we used the short-hand notation that $\underline{\sigma}_x$ is a
vector of operators $(\sigma^x,\sigma^y,\sigma^z)$ acting on site $x$.
Note that $Q_2$  is one half of the Hamiltonian.

Further charges can be computed easily using the boost operator
\eqref{boostdef}. After explicit computations we find
\begin{align}
       Q_3 &= \sum_{x=1}^L \left(- \frac{1}{2}\right)( \underline{\sigma}_x \times \underline{\sigma}_{x+1} ) \cdot \underline{\sigma}_{x+2} \\
    Q_4 &= \sum_{x=1}^L \Big( [ ( \underline{\sigma}_x \times \underline{\sigma}_{x+1} ) \times \underline{\sigma}_{x+2} ] \cdot \underline{\sigma}_{x+3} - 2 \underline{\sigma}_{x} \cdot \underline{\sigma}_{x+1} + \underline{\sigma}_{x} \cdot \underline{\sigma}_{x+2}+1 \Big),
\end{align}
where the cross denotes the vectorial cross product. Further examples and a
closed form result for all $Q_\alpha$ can be found in \cite{GM-higher-conserved-XXX}

Let us define the generalized current operators $J_\alpha^\beta$
according to the definition \eqref{Jab}, with the charge densities as
given above. Then the real space representations for the first few
currents are
\begin{align}
    J_2^2(x) = &\frac{1}{2} ( \underline{\sigma}_{x-1} \times \underline{\sigma}_{x} ) \cdot \underline{\sigma}_{x+1} \\
    J_3^2(x) = &- [ ( \underline{\sigma}_{x-1} \times \underline{\sigma}_{x} ) \times \underline{\sigma}_{x+1} ] \cdot \underline{\sigma}_{x+2} - \underline{\sigma}_{x} \cdot \underline{\sigma}_{x+1}+1 \\
    J_2^3(x) = &- \frac{1}{2} \Big( [ ( \underline{\sigma}_{x-2} \times \underline{\sigma}_{x-1} ) \times \underline{\sigma}_{x} ] \cdot \underline{\sigma}_{x+1} + [ ( \underline{\sigma}_{x-1} \times \underline{\sigma}_{x} ) \times \underline{\sigma}_{x+1} ] \cdot \underline{\sigma}_{x+2} \Big) + \\
    &+ \underline{\sigma}_{x} \cdot \underline{\sigma}_{x+1} + \underline{\sigma}_{x-1} \cdot \underline{\sigma}_{x} - \underline{\sigma}_{x-1} \cdot \underline{\sigma}_{x+1}-1. \nonumber
\end{align}
We can observe the equality $J_2^2(x)=-Q_3(x)$, in accordance
with the previous section. Similarly, we find the interesting relation
\begin{equation}
  Q_4=-\sum_{x=1}^L J_2^3(x), 
\end{equation}
which is an analogous identity that follows from \eqref{Jszimm}; the
minus sign is simply a matter of convention.

\section{The summation of the form factor expansion}

\label{sec:adjproof}

In order to sum up the expansion for the current mean values we write the r.h.s. of \eqref{fftcsa2J} in a slightly different way
\begin{equation}
\sum\limits_{\{\lambda^+\}\cup\{\lambda^-\}} \FJa(\{\lambda^+\}) \rho({\lambda^-})=
{\bf e'} \cdot A \cdot {\bf q}_\alpha,
\end{equation}
where $\FJa(\{\lambda\})$ are the symmetric diagonal form factors of
the currents and we defined
\begin{equation}
A_{jk}=\sum\limits_{\substack{{\{\lambda^+\}\cup\{\lambda^-\}}\\ \lambda_{j,k}\in\{\lambda^+\}}}\left(\sum_{\mathcal{T}}  \prod_{l_{pq}\in \mathcal{T}} \varphi_{pq}\right)\rho({\lambda^-}).
\end{equation}
Here the summation runs over $\mathcal{T}$, the non-directed spanning
trees of $\{\lambda^+\}$. Using this notation our task is to show that
$\sum_{l=1}^N G_{jl} A_{lk}=\delta_{jk}\cdot\det G$.

This matrix product can be written out explicitly using the matrix-tree theorem for the Gaudin determinant
\begin{equation}
\label{AGprod}
\begin{split}
& \sum_{l=1}^N\left(\left[ \delta_{jl}\big[p'_jL+\sum_{s=1}^N \varphi_{js}
    \big]-\varphi_{jl}\right] \times
  \sum_{\substack{\{\lambda^+\}\cup\{\lambda^-\}\\ \lambda_{k,l}\in\{\lambda^+\}}}\left(\sum_{\mathcal{T}}  \prod_{l_{pq}\in \mathcal{T}} \varphi_{pq}\right)
 \left( \sum_{\mathcal{F}} \prod_{n\in R(\mathcal{F})} p'_nL
  \prod_{l_{uv}\in \mathcal{F}} \varphi_{uv}\right) \right),
\end{split}
\end{equation}
where $\mathcal{F}$ denotes the directed spanning forests of
$\{\lambda^-\}$. First we show that the diagonal elements of this
matrix product give $\det G$. To do this we consider the $j=k$ case in
\eqref{AGprod} and we split the outer summation over $l$ into two parts: to
the $l=j$ and $l\neq j$ cases.

In the case of $l=j$ the sum of
the $\varphi$ terms from the Gaudin matrix appears, multiplying the appropriate
terms in $A$. These can be written in the following way: 
\begin{equation}
\begin{split}
& \left(\sum_{\substack{s=1\\s\neq j}}^N\varphi_{js}\right)\sum_{\substack{\{\lambda^+\}\cup\{\lambda^-\}\\ \lambda_j\in\{\lambda^+\}}}\left(\sum_{\mathcal{T}}  \prod_{l_{pq}\in \mathcal{T}} \varphi_{pq}\right)\left( \sum_{\mathcal{F}} \prod_{n\in R(\mathcal{F})} p'_nL
  \prod_{l_{uv}\in \mathcal{F}} \varphi_{uv}\right)=\\
& =\sum_{\substack{s=1\\ s\neq j}}\left[\varphi_{js}\sum_{\substack{\{\lambda^+\}\cup\{\lambda^-\}\\ \lambda_{j,s}\in \{\lambda^+\}}}\Big(\dots\Big)\Big(\dots\Big)+\varphi_{js}\sum_{\substack{\{\lambda^+\}\cup\{\lambda^-\}\\ \lambda_{j}\in \{\lambda^+\}\\ \lambda_{s}\in \{\lambda^-\}}}\Big(\dots\Big)\Big(\dots\Big)\right].
\end{split}
\end{equation}
Here in the second line we denoted the terms inside the summation with $(\dots)$ for brevity. Using this equation and renaming the summation variable $l$ to $s$ in \eqref{AGprod}, we can finally arrive to the following expression for the diagonal elements of the matrix product \eqref{AGprod}:
\begin{equation}
\label{AGprod_diag}
\begin{split}
& p_j'L\sum_{\substack{\{\lambda^+\}\cup\{\lambda^-\}\\ \lambda_j\in\{\lambda^+\}}}\left(\sum_{\mathcal{T}}  \prod_{l_{pq}\in \mathcal{T}} \varphi_{pq}\right)
 \left( \sum_{\mathcal{F}} \prod_{n\in R(\mathcal{F})} p'_nL
  \prod_{l_{uv}\in \mathcal{F}} \varphi_{uv}\right)\\
& +\sum_{\substack{s=1\\ s\neq j}}^N\left(\varphi_{js} 
\sum_{\substack{\{\lambda^+\}\cup\{\lambda^-\}\\ \lambda_j\in\{\lambda^+\} \\ \lambda_s\in\{\lambda^-\}}}\left(\sum_{\mathcal{T}}  \prod_{l_{pq}\in \mathcal{T}} \varphi_{pq}\right)
 \left( \sum_{\mathcal{F}} \prod_{n\in R(\mathcal{F})} p'_nL
  \prod_{l_{uv}\in \mathcal{F}} \varphi_{uv}\right)  \right).
\end{split}
\end{equation}
Looking at this expression from the graph theoretical point of view,
and using the matrix-tree theorem it can be shown that this is indeed
$\det G$: in the first term for every bi-partition the summation over
$\mathcal{T}$ (together with the $p_j'L$ factor) gives the
contribution of every such directed spanning tree of $\{\lambda^+\}$
in which $\lambda_j$ is the root. Together with the summation over
$\mathcal{F}$, which is just the contribution from every directed
spanning forest of $\{\lambda^-\}$, the first term contains all such
directed spanning forests of $\{\lambda\}$ in which $\lambda_j$ is one
of the roots. On the other hand in the second term the summation over
$\mathcal{T}$ gives the contributions from the free (rootless)
spanning trees of $\{\lambda^+\}$. These are ,,connected'' to one of
the spanning trees in the spanning forests of $\{\lambda^-\}$ by the
factor $\varphi_{js}$. Since $\{\lambda^+\}$ and $\{\lambda^-\}$ are
disjoint sets, this connection cannot create circles, so the result is
still a spanning tree. Furthermore, since we sum over $s$, this
connection is realized in every possible way, which means that the
second term gives the contribution of every spanning forest of
$\{\lambda\}$ in which $\lambda_j$ is not among the roots. Altogether
the two terms contain the contribution from all of the directed
spanning forests of $\{\lambda\}$, which is just $\det G$ according to
the matrix-tree theorem.

For the off-diagonal elements of the matrix
product \eqref{AGprod} the same steps can be performed. To show that
in this case the result is zero, it is convenient to investigate  those terms
that contain $p'_j$, and afterwards those that do not. The ones that
contain $p_j'$ are 
\begin{equation}
\begin{split}
& p_j'L \sum_{\substack{\{\lambda^+\}\cup\{\lambda^-\}\\ \lambda_{j,k}\in\{\lambda^+\}}}\left(\sum_{\mathcal{T}}  \prod_{l_{pq}\in \mathcal{T}} \varphi_{pq}\right)
 \left( \sum_{\mathcal{F}} \prod_{n\in R(\mathcal{F})} p'_nL
  \prod_{l_{uv}\in \mathcal{F}} \varphi_{uv}\right)- \\
& -\sum_{\substack{s=1\\ s\neq j}}^N\left(\varphi_{js} 
\sum_{\substack{\{\lambda^+\}\cup\{\lambda^-\}\\ \lambda_{k,s}\in\{\lambda^+\} \\ \lambda_j\in\{\lambda^-\}}}\left(\sum_{\mathcal{T}}  \prod_{l_{pq}\in \mathcal{T}} \varphi_{pq}\right)
 \left( \sum_{\mathcal{F}} \prod_{\substack{n\in R(\mathcal{F})\\j\in R(\mathcal{F})}} p'_nL
  \prod_{l_{uv}\in \mathcal{F}} \varphi_{uv}\right)  \right).
\end{split}
\end{equation}

In the second line $j\in R(\mathcal{F})$ denotes that here we consider
only those spanning trees of $\{\lambda^-\}$ in which $\lambda_j$ is
one of the roots. Because of this we can pull out a $p_j'L$ factor
from it. Similarly to the previous argument it can be shown that this
whole expression is zero: in the second line, for every partition the
spanning tree in $\mathcal{F}$ that originates from $\lambda_j$ is
connected to a spanning tree in $\mathcal{T}$ by $\varphi_{js}$, and
the result of this connection is still a spanning tree. Since there is
a summation over $s$, all possible connections are included. This
means that the second line contains all such spanning trees that
include $\lambda_j$ and $\lambda_k$, beside all the spanning forests
of $\{\lambda^-\}$. But the first line is exactly the same, so their
difference is zero. The terms that do not contain $p_j'$ are 
\begin{equation}
\begin{split}
& \sum_{\substack{s=1\\ s\neq j}}^N \left[\varphi_{js} 
\sum_{\substack{\{\lambda^+\}\cup\{\lambda^-\}\\ \lambda_{j,k}\in\{\lambda^+\} \\ \lambda_s\in\{\lambda^-\}}}\left(\sum_{\mathcal{T}}  \prod_{l_{pq}\in \mathcal{T}} \varphi_{pq}\right)
 \left( \sum_{\mathcal{F}} \prod_{n\in R(\mathcal{F})} p'_nL
  \prod_{l_{uv}\in \mathcal{F}} \varphi_{uv}\right) - \right. \\
& \left.-\varphi_{js} 
\sum_{\substack{\{\lambda^+\}\cup\{\lambda^-\}\\ \lambda_{k,s}\in\{\lambda^+\} \\ \lambda_j\in\{\lambda^-\}}}\left(\sum_{\mathcal{T}}  \prod_{l_{pq}\in \mathcal{T}} \varphi_{pq}\right)
 \left( \sum_{\mathcal{F}} \prod_{\substack{n\in R(\mathcal{F})\\j\not\in R(\mathcal{F})}} p'_nL
  \prod_{l_{uv}\in \mathcal{F}} \varphi_{uv}\right) \right].
\end{split}
\end{equation}
Here in the first line $\lambda_j \in \{\lambda^+\}$ and $\lambda_s
\in \{\lambda^-\}$, and they are connected by $\varphi_{js}$, while in
the second line they are reversed. Since we sum up for every $s$ and
every bi-partition (and since $\lambda_j$ cannot be a root in the
second line) the difference is zero. This completes the summation of
the expansion. 

\section{Proof of the singularity properties}

\label{sec:singproof}

Here we present the proof of the singularity property
\eqref{sing_prop}, focusing on the case of the XXZ chain. Following the reasoning in Section \ref{sec:fftcsa2}
it is sufficient to show that \eqref{sing_prop} holds for an operator
which is the product of the elements of the monodromy matrix:
\begin{equation}
  \label{xet}
\ordo=X(\mu_1)\dots X(\mu_M),
\end{equation}
where $X(\mu)$ represents one of the
four operators: $A,B,C,D$. In order to prove \eqref{sing_prop} we are
using the commutation relations of the operators $A,B,C,D$, and the
following singularity property of the scalar products (see Subsection
IX. 3. of \cite{Korepin-Book})
\begin{equation}
\label{sing_prop_scalar}
\bra{0}\prod_{j=1}^N \mathbb{C}(\lambda^C_j)\prod_{j=1}^N
\mathbb{B}(\lambda^B_j)\ket{0}
\xrightarrow{\lambda_N^C\rightarrow\lambda_N^B}\frac{i\sinh(\eta)}{\lambda_N^C-\lambda_N^B}\big(l_N^C-l_N^B\big)
\prod_{k=1}^{N-1}f_{Nk}^Cf_{Nk}^B
\bra{0}\prod_{j=1}^{N-1} \mathbb{C}(\lambda^C_j)\prod_{j=1}^{N-1} \mathbb{B}(\lambda^B_j)\ket{0}^{mod}.
\end{equation} 
Here the elements of the sets $\{\lambda^C\}$ and $\{\lambda^B\}$ not
necessarily satisfy the Bethe equations, and the $\vev{\ }^{mod}$
notation means, that the scalar product is calculated with the
modified vacuum expectation values $a^{mod}$ and $d^{mod}$ defined in
\eqref{admod}.

We stress an important technical detail already at this point: whereas
for the Bethe
states we used the re-normalized creation operators
$\mathbb{B}(\lambda)$ and $\mathbb{C}(\lambda)$, because these are
convenient to study the norms and mean values, for the local
operator $\ordo$ we used the original operators
$A(\lambda),B(\lambda),C(\lambda),D(\lambda)$, because these are
needed for the solution of the inverse problem.

For the sake of completeness we also present here the commutation relations
resulting from the RTT equation \eqref{eq:RTT}:
\begin{equation}
\begin{aligned}
A(\mu)B(\lambda)&=f(\lambda,\mu)B(\lambda)A(\mu)+g(\mu,\lambda)B(\mu)A(\lambda)\\
B(\mu)A(\lambda)&=f(\lambda,\mu)A(\lambda)B(\mu)+g(\mu,\lambda)A(\mu)B(\lambda)\\
D(\lambda)B(\mu)&=f(\lambda,\mu)B(\mu)D(\lambda)+g(\mu,\lambda)B(\lambda)D(\mu)\\
B(\lambda)D(\mu)&=f(\lambda,\mu)D(\mu)B(\lambda)+g(\mu,\lambda)D(\lambda)B(\mu)\\
A(\lambda)C(\mu)&=f(\lambda,\mu)C(\mu)A(\lambda)+g(\mu,\lambda)C(\lambda)A(\mu)\\
C(\lambda)A(\mu)&=f(\lambda,\mu)A(\mu)C(\lambda)+g(\mu,\lambda)C(\mu)A(\lambda)\\
D(\mu)C(\lambda)&=f(\lambda,\mu)C(\lambda)D(\mu)+g(\mu,\lambda)C(\mu)D(\lambda)\\
C(\mu)D(\lambda)&=f(\lambda,\mu)D(\lambda)C(\mu)+g(\mu,\lambda)D(\mu)C(\lambda)\\
[A(\lambda),D(\mu)]&=g(\lambda,\mu)\{C(\mu)B(\lambda)-C(\lambda)B(\mu)\}\\
[D(\lambda),A(\mu)]&=g(\lambda,\mu)\{B(\mu)C(\lambda)-B(\lambda)C(\mu)\}\\
[B(\lambda),C(\mu)]&=g(\lambda,\mu)\{D(\mu)A(\lambda)-D(\lambda)A(\mu)\}\\
[C(\lambda),B(\mu)]&=g(\lambda,\mu)\{A(\mu)D(\lambda)-A(\lambda)D(\mu)\}\\
[A(\lambda),A(\mu)]&=[B(\lambda),B(\mu)]=[C(\lambda),C(\mu)]=[D(\lambda),D(\mu)]=0,
\end{aligned}
\label{komm_2}
\end{equation}
where
\begin{equation}
  f(u,v)=\frac{\sin(u-v+i\eta)}{\sin(u-v)},\qquad g(u,v)=\frac{i\sinh(\eta)}{\sin(u-v)}.
\end{equation}

Instead of considering an arbitrary product in \eqref{xet} it is useful to require a specific ordering of the
$X=A,B,C,D$ operators.
It follows from the commutation relations above that any
operator of the form $X(\mu_1)\dots X(\mu_M)$ can be written as a
linear combination of operators in which the order of the $A, B, C$
and $D$ operators is fixed in the following way: 
\begin{equation}
X(\mu_1)\dots X(\mu_M)=\sum G(\{\mu\})CC\dots CDD\dots DAA\dots ABB\dots B,
\end{equation}
where the $G(\{\mu\})$ are the coefficients of the individual terms,
which will include combinations of the $f(u)$ and $g(u)$ functions.
Because of the linearity of the scalar product, all we need to show is
that \eqref{sing_prop} holds for a single term in this
summation.

We will first consider
products of the form $\ordo=D(\mu_1)D(\mu_2)\dots
D(\mu_m)A(\mu_{m+1})\dots A(\mu_M)$. The possible presence of
additional $C$ and $B$ operators will be treated later.

First let us consider only one $A$ operator:
\begin{equation}
\begin{split}
\bra{0}\prod_{k=1}^N \mathbb{C}(\lambda_k^C)A(\mu)\prod_{k=1}^N \mathbb{B}(\lambda_k^B)\ket{0}&=a(\mu)\prod_{k=1}^Nf(\lambda_k^B,\mu)\bra{0}\prod_{k=1}^N \mathbb{C}(\lambda_k^C)\prod_{k=1}^N \mathbb{B}(\lambda_k^B)\ket{0}+\\
&+\sum_{n=1}^N a(\lambda_n^B) g(\mu,\lambda_n^B) \prod_{\substack{k=1\\k\neq n}}^N f(\lambda_k^B,\lambda_n^B)\frac{d(\mu)}{d(\lambda_n^B)} \bra{0}\prod_{k=1}^N \mathbb{C}(\lambda_k^C) \mathbb{B}(\mu)\prod_{\substack{k=1\\k\neq n}}^N \mathbb{B}(\lambda_k^B)\ket{0}.
\end{split}
\end{equation}
Here we only used the commutation relation between the operators $A$
and $B$. Taking the $\lambda_N^C\rightarrow\lambda_N^B$ limit and
using \eqref{sing_prop_scalar} we arrive at the expected expression: 
\begin{equation}
\begin{split}
\bra{0}\prod_{k=1}^N \mathbb{C}(\lambda_k^C)A(\mu)\prod_{k=1}^N \mathbb{B}(\lambda_k^B)\ket{0}\xrightarrow{\lambda_N^C\rightarrow \lambda_N^B}
\frac{i\sinh(\eta)}{\lambda_N^C-\lambda_N^B}\big(l_N^C-l_N^B\big)\prod_{k=1}^{N-1}f_{Nk}^Cf_{Nk}^B\bra{0}\prod_{j=1}^{N-1} \mathbb{C}(\lambda^C_j)A(\mu)\prod_{j=1}^{N-1} \mathbb{B}(\lambda^B_j)\ket{0}^{mod}.
\end{split}
\end{equation}
The computation goes similarly for the $D$ operator. Now let's consider the case when $\ordo=\prod_{k=1}^M A(\mu_k)$. To calculate the matrix element of this operator, first we have to compute the effect of it on an arbitrary state: 
\begin{equation}
\prod_{l=1}^MA(\mu_l)\prod_{k=1}^N\mathbb{B}(\lambda_k)\ket{0}.
\end{equation}
It is obvious that the result is the linear combination of states with particle number $N$. The rapidities of the particles are coming from the set $\{\mu\}_M\cup\{\lambda\}_N$. The result can be arranged into a sum, depending on the number of rapidities coming from $\{\mu\}_M$. In every term in this summation we have to take into consideration every possible way how that certain amount of rapidities can be substituted from $\{\mu\}_M$. All this means that we can write the result in the following way:
\begin{equation}
\label{prod_AB}
\begin{split}
\prod_{l=1}^MA(\mu_l)\prod_{k=1}^N\mathbb{B}(\lambda_k)\ket{0}=
\sum_{K=0}^{\max(N,M)}
\sum_KG_{l_1,l_2,\dots,l_K}^{n_1,n_2,\dots,n_K}(\{\lambda\}_N|\{\mu\}_M)\prod_{q=1}^K\mathbb{B}(\mu_{lq})\!\!\!\!\!\!\!\!\!\prod_{\substack{k=1\\k\neq n_1,\dots,n_K}}^N\!\!\!\!\!\!\!\!\!\mathbb{B}(\lambda_k)\ket{0},
\end{split}
\end{equation}
where $\sum_K$ goes over every such subset of $\{\lambda\}_N\cup\{\mu\}_M$ that has the following two properties: the number of its elements is $N$ and it has exactly $K$ elements coming from $\{\mu\}_M$. If $\mu_{l_1},\mu_{l_2},\dots,\mu_{l_K}$ denotes these elements and $\lambda_{n_1},\lambda_{n_2},\dots,\lambda_{n_K}$ $\{\lambda\}_N$ those that are not present in the subset, $\sum_K$ can be written out explicitly:
\begin{equation}
\sum_K=\sum_{l_1=1}^M\sum_{\substack{l_2=1\\l_2\neq l_1}}^M\dots\!\!\!\!\sum_{\substack{l_K=1\\l_K\neq l_1,\dots,l_{K-1}}}^M\sum_{n_1=1}^N\sum_{\substack{n_2=1\\n_2\neq n_1}}^N\dots\!\!\!\!\sum_{\substack{n_K=1\\n_K\neq n_1,\dots,n_{K-1}}}^N.
\end{equation}
The $G$ coefficients can be calculated using only the commutation relations. For example in the case $K=0$ none of the rapidities are coming from $\{\mu\}_M$ and we have:
\begin{equation}
G^0_0(\{\lambda\}_N|\{\mu\}_M)=\prod_{l=1}^M\left(a(\mu_l)\prod_{k=1}^Nf(\lambda_k,\mu_l)\right).
\end{equation}
Let's calculate now a general coefficient, where $K$ rapidities are
coming from $\{\mu\}_M$! To do this we have to keep in mind that every
time when we exchange the arguments of the operators $A(\mu)$ and
$\mathbb{B}(\lambda)$ during commutation there will appear a factor
$g(\mu,\lambda)d(\mu)/d(\lambda)$, while if we do not exchange them
the corresponding factor will be $f(\lambda,\mu)$.

Let $\mu_{l_1}$ be the first substituted rapidity and $\lambda_{n_1}$ the one which is replaced! In this case there will be a $g(\mu_{l_1},\lambda_{n_1})d(\mu_{l_1})/d(\lambda_{n_1})$ factor coming from the commutation relation of $A(\mu_{l_1})$ and $\mathbb{B}(\lambda_{n_1})$, describing the replacement. The commutation of $A(\lambda_{n_1})$ and the other $\mathbb{B}$ operators will give the factor $\prod_{\substack{k=1\\k\neq n_1}}f(\lambda_{k},\lambda_{n_1})$. Finally the effect of $A(\lambda_{n_1})$ on the pseudo vacuum will result in the factor $a(\lambda_{n_1})$.
In the case of the next substituted rapidity ($\mu_{l_2}$ replacing
$\lambda_{n_2}$) the same factors will appear, expect that now it has
to be taken into consideration that one of the commutations will be
with $\mathbb{B}(\mu_{l_1})$ and not with
$\mathbb{B}(\lambda_{n_1})$.

The rest of the substitutions go the same way. The remaining $A$
operators with the non-substituted rapidities commute through the
$\mathbb{B}$ operators without exchanging the arguments. To take into
consideration all of the possible cases, we have to sum over every $n$
and $l$. But since both the $A$ and $\mathbb{B}$ operators commute
with each other, the answer does not depend on which $\lambda$ is replaced by which $\mu$. Therefore, to obtain the right result we have to divide by $K!$. This leads to the following result for a generic coefficient:
\begin{equation}
\begin{split}
&G_{l_1,l_2,\dots,l_K}^{n_1,n_2,\dots,n_K}(\{\lambda\}_N|\{\mu\}_M)=\frac{1}{K!}a(\lambda_{n_1})g(\mu_{l_1},\lambda_{n_1})\frac{d(\mu_{l_1})}{d(\lambda_{n_1})}\left(\prod_{\substack{k=1\\k\neq n_1}}^N f(\lambda_{k},\lambda_{n_1})\right) a(\lambda_{n_2})g(\mu_{l_2},\lambda_{n_2})\frac{d(\mu_{l_2})}{d(\lambda_{n_2})}\times\\ 
&\times \left(\prod_{\substack{k=1\\k\neq n_1,n_2}}^N f(\lambda_k,\lambda_{n_2})\right)f(\mu_{l_1},\lambda_{n_2})\ \dots a(\lambda_{n_K})g(\mu_{l_K},\lambda_{n_K})\frac{d(\mu_{l_K})}{d(\lambda_{n_K})}\left(\prod_{\substack{k=1\\k\neq n_1,\dots,n_{K-1}}}^N\!\!\!\!\!\!\!\!f(\lambda_{k},\lambda_{n_K})\right)\times\\
&\times\left(\prod_p^{K-1}f(\mu_{l_p},\lambda_{n_K})\right)
\prod_{\substack{r=1\\r\neq l_1,\dots,l_K}}^M\left(a(\mu_r)\!\!\!\!\!\prod_{\substack{k=1\\k\neq n_1,\dots,n_K}}^N\!\!\!\!\!f(\lambda_k,\mu_r)\prod_{q\in\{l_1,\dots,l_K\}}f(\mu_q,\mu_r)\right).
\end{split}
\end{equation}
Multiplying \eqref{prod_AB} from the left with $\bra{0}\prod_{k=1}^N\mathbb{C}(\lambda_k^C)$ and taking the $\lambda_N^C\rightarrow\lambda_N^B$ limit, in the r.h.s. only those terms will give a contribution to the pole, in which $\lambda_N^B$ is still among the arguments of the $\mathbb{B}$ operators. All these terms will get multiplied by $\frac{i\sinh(\eta)}{\lambda_N^C-\lambda_N^B}\big(l_N^C-l_N^B\big)\prod_{k=1}^{N-1}f_{Nk}^Cf_{Nk}^B$ according to \eqref{sing_prop_scalar}, and the modified scalar products will appear with $N-1$ particles. Since in every such $G(\{\lambda\}_N|\{\mu\}_M)$ coefficient where $N\not\in\{n_1,\dots,n_K\}$ there is a $f(\lambda_N^B,\xi)$ factor next to every vacuum expectation value $a(\xi)$, it is true that:
\begin{equation}
G_{l_1,l_2,\dots,l_K}^{n_1,n_2,\dots,n_K}(\{\lambda\}_N|\{\mu\}_M)=G_{l_1,l_2,\dots,l_K}^{n_1,n_2,\dots,n_K}(\{\lambda\}_{N-1}|\{\mu\}_M)^{mod}\qquad\qquad (N\not\in\{n_1,\dots,n_K\}),
\end{equation} 
which means that \eqref{sing_prop} holds for $\ordo=\prod_{k=1}^M
A(\mu_k)$. This computation goes the same way for the product of $D$
operators (only the order of the arguments of the $g$ and $f$
functions are reversed). But because of the linearity of the scalar
product this also proves that \eqref{sing_prop} is true for an
operator of the form $\ordo=D(\mu_1)D(\mu_2)\dots
D(\mu_m)A(\mu_{m+1})\dots A(\mu_M)$: the effect of the products of the
$A$ operators on an arbitrary state is a linear combination of states
computed above. The effect of the product of the $D$ operators on each
term in this linear combination is another linear combination, with
coefficients that can be similarly calculated. However, if the
original term (the one obtained after acting with the $A$ operators)
doesn't contain $\lambda_N^B$ among the arguments of the $\mathbb{B}$
operators then it won't appear there after acting with the $D$
operators. If it is still an argument of a $\mathbb{B}$ operator in
the original term, than the appropriate $f$ factor is present next to
the v.e.v., and it will appear also after acting with the $D$
operators. This means that in every term that contains
$\mathbb{B}(\lambda_N^B)$ the appropriate $f(\lambda_N,\xi)$ or
$f(\xi,\lambda_N)$ factor will appear next to the v.e.v. $a(\xi)$ or
$d(\xi)$. Therefore \eqref{sing_prop} holds for
$\ordo=D(\mu_1)D(\mu_2)\dots D(\mu_m)A(\mu_{m+1})\dots A(\mu_M)$.

To complete the proof, we consider now the case when there are
additional $C$ and $B$ operators. Since they are on the left and right
hand side, respectively, they can be considered as additional creation
operators in the states. Their presence does not alter the singularity
property \eqref{sing_prop} we intend to prove. To see this, let us
first consider the most simple case when $\ordo=C(\mu^C)B(\mu^B)$. The
singularities of the form factors of this operator follow 
from \eqref{sing_prop_scalar}. First we consider the re-normalized operators:
\begin{equation}
  \label{sing_prop_scalar2}
  \begin{split}
&\bra{0}\left(\prod_{j=1}^N \mathbb{C}(\lambda^C_j)\right)
\mathbb{C}(\mu^C)\mathbb{B}(\mu^B)
\left(\prod_{j=1}^N \mathbb{B}(\lambda^B_j)\right)\ket{0}
\xrightarrow{\lambda_N^C\rightarrow\lambda_N^B}\\
&\frac{i\sinh(\eta)}{\lambda_N^C-\lambda_N^B}\big(l_N^C-l_N^B\big)
\left(\prod_{k=1}^{N-1}f_{Nk}^Cf_{Nk}^B\right) f(\lambda_N,\mu^C) f(\lambda_N,\mu^B)
  \bra{0}
\left(\prod_{j=1}^{N-1}\mathbb{C}(\lambda^C_j)\right)
\mathbb{C}(\mu^C)\mathbb{B}(\mu^B)
\left(\prod_{j=1}^{N-1}\mathbb{B}(\lambda^B_j)\right)\ket{0}^{mod}.
  \end{split}
\end{equation} 
Substituting the definition \eqref{renorm} and also using the
modification rule \eqref{admod} for the re-normalization on the
r.h.s. we get
\begin{equation}
  \label{sing_prop_scalar3}
  \begin{split}
\bra{0}\left(\prod_{j=1}^N \mathbb{C}(\lambda^C_j)\right)
&C(\mu^C)B(\mu^B)
\left(\prod_{j=1}^N \mathbb{B}(\lambda^B_j)\right)\ket{0}
\xrightarrow{\lambda_N^C\rightarrow\lambda_N^B}\\
&\frac{i\sinh(\eta)}{\lambda_N^C-\lambda_N^B}\big(l_N^C-l_N^B\big)
\left(\prod_{k=1}^{N-1}f_{Nk}^Cf_{Nk}^B\right) 
  \bra{0}
\left(\prod_{j=1}^{N-1}\mathbb{C}(\lambda^C_j)\right)
C(\mu^C)B(\mu^B)
\left(\prod_{j=1}^{N-1}\mathbb{B}(\lambda^B_j)\right)\ket{0}^{mod},
  \end{split}
\end{equation}
which has just the desired form. It can be argued similarly, that the
presence of the $C$ and $B$ operators does not alter the singularity
properties when there are $D$ and $A$ operators present.

\end{widetext}

\addcontentsline{toc}{section}{References}
%


\end{document}